\documentclass[journal]{IEEEtran}
\usepackage{array}
\usepackage{graphics}
\usepackage{cite}
\usepackage[pdftex]{graphicx}
\usepackage{epstopdf}
\usepackage{epsfig}
\usepackage{latexsym}
\usepackage{amsfonts}
\usepackage{calc}
\usepackage{url}
\usepackage{enumerate}
\usepackage{color}
\usepackage[tbtags]{amsmath}
\usepackage{amssymb}
\usepackage{upref}
\usepackage{dsfont}
\usepackage{multirow}
\usepackage{booktabs}
\usepackage{bigstrut}
\usepackage{rotating}
\usepackage{mathtools,bbm}
\usepackage{arydshln}
\usepackage{breqn}
\usepackage[acronym]{glossaries}
\usepackage{multirow}
\usepackage{longtable}
\usepackage{lipsum}
\usepackage{tikz}
\usepackage{callouts}
\usepackage{siunitx}

\makeglossaries

\newacronym{dnf}{DNF}{disjunctive normal form}

\newtheorem{theorem}{Theorem}

\newtheorem{lemma}{Lemma}

\newcommand{\qed}{\hfill $\square$}

\begin{document}

\title{A Fast and Accurate Failure Frequency Approximation for $k$-Terminal Reliability Systems}


\author{Anoosheh Heidarzadeh,~\IEEEmembership{Member,~IEEE,}
Alex Sprintson,~\IEEEmembership{Senior Member,~IEEE,} and Chanan Singh~\IEEEmembership{Fellow,~IEEE}
\thanks{The authors are with the Department of Electrical and Computer Engineering, Texas A\&M University, College Station, TX 77843 (E-mail: \{anoosheh,spalex,singh\}@tamu.edu).}}

\maketitle
\thispagestyle{plain}
\begin{abstract}
This paper considers the problem of approximating the failure frequency of large-scale composite $\boldsymbol{k}$-terminal reliability systems. In such systems, the nodes ($\boldsymbol{k}$ of which are terminals) are connected through components which are subject to random failure and repair processes. At any time, a system failure occurs if the surviving system fails to connect all the $\boldsymbol{k}$ terminals together. We assume that each component's up-times and down-times follow statistically independent stationary random processes, and these processes are statistically independent across the components. In this setting, the exact computation of failure frequency is known to be computationally intractable (NP-hard). In this work, we present an algorithm to approximate the failure frequency for any given multiplicative error factor that runs in polynomial time in the number of (minimal) cutsets. Moreover, for the special case of all-terminal reliability systems, i.e., where all nodes are terminals, we propose an algorithm for approximating the failure frequency within an arbitrary multiplicative error that runs in polynomial time in the number of nodes (which can be much smaller than the number of cutsets). In addition, our simulation results confirm that the proposed method is much faster and more accurate than the Monte Carlo simulation technique for approximating the failure frequency. 
\end{abstract}
 




\section*{Acronyms}
\vspace{0.125cm}
\hspace{-0.675cm}
\begin{tabular}{p{1.5cm}p{6.95cm}}
DNF & disjunctive normal form\\
KLM & Karp-Luby-Madras\\
MCS & Monte Carlo simulation\\
NP-hard & non-deterministic polynomial-time hard\\
RGC & recursive generalized contraction
\end{tabular}

\section*{Notation}
\vspace{0.125cm}
\hspace{-0.675cm}
\begin{tabular}{p{1.5cm}p{6.95cm}}
$\Psi$ & a composite system\\
$m$ & number of components\\
$n$ & number of nodes\\
$k$ & number of terminal nodes\\
$[i]$ & set of integers $\{1,\dots,i\}$\\
$\lambda_i$ & failure rate of component $i$\\
$\mu_i$ & repair rate of component $i$\\
$\lambda_{\text{max}}$ & maximum failure rate of a component\\
$\mu_{\text{min}}$ & minimum repair rate of a component\\
$\lambda$ & sum of failure rates of all components\\
$\mu$ & sum of repair rates of all components\\
$p_i$ & probability of component $i$ being unavailable\\
$w_i$ & weight of component $i$\\
$w_{\text{max}}$ & maximum weight of a component\\
$w$ & sum of weights of all components\\
\end{tabular}


\vspace{0.125cm}
\hspace{-0.675cm}
\begin{tabular}{p{1.5cm}p{6.95cm}}

$C$ & set of all minimal cutsets\\
$N$ & number of minimal cutsets\\
$\mathcal{C}_j$ & $j$th minimal cutset\\
$s^{*}$ & minimum size of a cutset\\
$w(\mathcal{C})$ & weight of a cutset $\mathcal{C}$\\
$w^{*}$ & minimum weight of a cutset\\
$p^{*}$ & maximum failure probability of a cutset\\
$p(\mathcal{C}_{\mathcal{I}})$ & probability of all components in the collection of cutsets $\{\mathcal{C}_j\}_{j\in \mathcal{I}}$ being unavailable\\
$\epsilon$ & target approximation error factor\\
$\delta$ & target approximation error probability\\
$P_f$ & probability of system failure in steady-state\\
$P_f^{+}, P_f^{-}$ & first-order upper- and lower-bound on $P_f$\\
$\hat{P}_f$ & an approximation of $P_f$ using bounding technique\\
$\tilde{P}_f$ & an approximation of $P_f$\\
$F_f$ & frequency of system failure in steady-state\\
$F_f^{+}, F_f^{-}$ & first-order upper- and lower-bound on $F_f$\\
$\hat{F}_f$ & an approximation of $F_f$ using bounding technique\\
$\tilde{F}_f$ & an approximation of $F_f$\\
$P$ & probability of all components of some cutset being unavailable and unexposed\\
$\tilde{P}$ & an approximation of $P$\\
$C^{(\alpha)}$ & set of all $\alpha$-min cutsets for arbitrary $\alpha\geq 1$\\
$N^{(\alpha)}$ & number of $\alpha$-min cutsets\\
$P^{(\alpha)}$ & probability of all components of some $\alpha$-min cutset being unavailable and unexposed \\
$P_f^{(\alpha)}$ & probability of all components of some $\alpha$-min cutset being unavailable\\
$\tilde{P}_f^{\text{MC}}$ & an approximation of $P_f$ using Monte Carlo simulation\\
$\tilde{F}_f^{\text{MC}}$ & an approximation of $F_f$ using Monte Carlo simulation\\
$\boldsymbol{s}$ & (Boolean) system-state vector\\
$\mathcal{S}_f$ & set of all system-states in which the system is unavailable\\
$\mathcal{I}(\boldsymbol{s})$ & set of unavailable components in system-state $\boldsymbol{s}$\\
$p(\boldsymbol{s})$ & probability of system-state $\boldsymbol{s}$\\
poly$(N)$ & a polynomial function in $N$\\
poly$(n)$ & a polynomial function in $n$\\
\end{tabular}

\section*{Nomenclature}
\vspace{0.125cm}
\hspace{-0.675cm}
\begin{tabular}{p{1.5cm}p{6.95cm}}
$\alpha$-min cutset & a minimal cutset of weight no greater than $\alpha$ times the minimum cutset weight\\
$(\epsilon,\delta)$-approx. & a multiplicative approximation with error factor at most $\epsilon$ and error probability at most $\delta$\\	
\end{tabular}


\section{Introduction}
Consider a composite system whose $n$ nodes, consisting of $k$ ($2\leq k\leq n$) \emph{terminals} and $n-k$ \emph{relays}, are connected through components that are subject to statistically-independent continuous-time stationary failure/repair random processes. At any given time, each component is either operational or not, and the system fails if the surviving system of operational components does not connect all terminals. The \emph{probability of failure} ($P_f$) and the \emph{frequency of failure} ($F_f$) are two important measures of reliability of such systems~\cite{GMP:64,SB:77,HTL:81}. These quantities are very useful to derive other reliability measures such as mean down-time and mean cycle-time~\cite{SB:77}.


Numerous algorithms were previously designed for computing $P_f$~\cite{BB:68,SB:77,A:79,S:81,F:86,L:87,BS:87,DS:88,H:89,W:90,LW:92,LJY:95,LP:01,BT:03} and $F_f$~\cite{SB:74,SB:77,S:79,S:80,S:81,S2:81,DS:88,B:82}. The proposed algorithms, however, become intractable in large-scale systems because the computational complexity grows very quickly with the number of nodes in the system (i.e., $n$). Specifically, the exact computation of these quantities was shown to be NP-hard~\cite{PB:83,B:86}. To overcome this challenge, various methods were developed to approximate $P_f$ and $F_f$. 


The existing techniques for approximating $P_f$ are based on the Monte Carlo simulation (MCS)~\cite{MPL:94,LCFA:02,SKS:02,HBKK:05,ADS:09} and rare-events simulation~\cite{JB:69,NBB:70,EP:70,WK:11}. Similar methods were used to approximate $F_f$ in~\cite{S:77} and~\cite{MS:99}. Notwithstanding, the only existing technique which can provably approximate $P_f$ within an arbitrary multiplicative error and runs in polynomial time in the number of cutsets (i.e., minimal collections of components whose failure results in a loss of connectivity of some terminals from the rest) was proposed in~\cite{K:01}. To our knowledge, no such computationally efficient algorithm with provable guarantees was previously proposed for approximating $F_f$ with an arbitrary multiplicative error factor. 



The algorithm of~\cite{K:01} estimates $P_f$ for the settings where the failure/repair random process for each component is stationary. In such settings, each component is available or unavailable at any time instant, independently from other components, with some constant probability (independent of time). Then, $P_f$ is equal to the probability that all components in a cutset are unavailable. The algorithm of~\cite{K:01} relies on the fact that $P_f$ can be written as the truth probability of some disjunctive normal form (DNF) formula, and a multiplicative approximation of this probability can be computed in polynomial time in the number of cutsets in the system. However, the number of cutsets can be exponential in $n$, and the application of this technique is not practical in such cases. 




Interestingly, for the special case where all $n$ nodes are terminals (i.e., $k=n$), it was shown~\cite{K:01} that a multiplicative approximation of $P_f$ can be computed in polynomial time in $n$. (Such results do not exist for more general cases where only a subset of nodes are terminals.) One example of such all-terminal reliability systems is the electricity distribution networks where the terminals represent either a supply point or a major load point. The feeders between these terminals are the components that can fail and be repaired. The loss of supply to any major load point is considered a system failure. At the distribution level there may be only one feed point for many of such networks, and in such cases, the networks are all-terminal reliability systems. (For a network with more than one feed point, depending on its topology, the network may or may not be an all-terminal reliability system.) Other examples of such systems can be found in communication networks, computer networks, and transportation networks~\cite{C:87}.

The main ideas in \cite{K:01} can be summarized as follows: (i) the number of \emph{weak cutsets} in a system, i.e., those cutsets with higher probability of failure, is polynomial in $n$, and the enumeration of all such cutsets can be done in polynomial time in $n$, and (ii) the probability that all components of a weak cutset are unavailable provides a multiplicative approximation of $P_f$. The idea of using weak cutsets was also used recently in~\cite{SM:09} for computing a bounding-type approximation of $P_f$. Nevertheless, there is no apparent connection between $F_f$ and the truth probability of a DNF formula. This implies the need for a novel technique for approximating $F_f$.

The main contributions of this work are as follows:
\begin{itemize}
\item We present a new algorithm that runs in polynomial time in the number of cutsets in the system, and approximates $F_f$ for $k$-terminal reliability systems (for any $k$) within any given multiplicative error. Moreover, for the special case where all nodes are terminals, we propose a new algorithm to compute an approximation of $F_f$ for any given multiplicative error factor that runs in polynomial time in $n$. To our knowledge, this is the first and only fast (polynomial time in $n$) and accurate (with arbitrary provable guarantees) algorithm for approximating $F_f$ for all-terminal reliability systems. 
\item We adapt the machinery of \cite{K:01}, which was tailored to approximate $P_f$, in a non-trivial way to approximate $F_f$. In particular, we present a new transformation to obtain $F_f$ from $P_f$ by introducing an auxiliary probability $P$, and re-writing $F_f$ as a scalar multiple of the difference between $P_f$ and $P$. To the best of our knowledge, this connection was not previously reported in the literature. This transformation enables us to convert the problem of approximating $F_f$ to the two sub-problems of approximating $P_f$ and $P$, each of which can be linked to a DNF formula via a carefully designed random process. 
\item It is well known that $P_f$ can be thought of as the probability that, under a random \emph{sampling process}, all components of a cutset in the system are unavailable. To relate $P$ to a random process, we define a new random process, referred to as the \emph{exposure process}, statistically independent from the sampling process, such that $P$ can be thought of as the probability that, under the sampling and exposure processes simultaneously, all components of a cutset in the system are unavailable and unexposed. 
\item We identify the weak cutsets of a system that satisfy the following property: under the sampling and exposure processes simultaneously, the probability that all components of a weak cutset are unavailable and unexposed provides multiplicative approximations of $P_f$ and $P$, and subsequently, a multiplicative approximation of $F_f$ for any arbitrary error factor. We also prove that the set of weak cutsets required for approximating $F_f$ contains (and is greater, but not more than a factor polynomial in $n$, than) the set of weak cutsets previously identified in \cite{K:01} for approximating $P_f$.
\item We show that MCS can provide an additive approximation of $F_f$ in polynomial time; whereas, using MCS, $F_f$ cannot be approximated in polynomial time within a multiplicative error. This suggests that approximating $F_f$ with a multiplicative error factor, when compared to approximating $F_f$ with an additive error factor, is computationally more expensive, and hence more challenging.
\item We compare the proposed technique and the MCS technique via simulations for approximating $F_f$  for a $3\times 3$ ($9$ nodes) grid network and for the layer 3 ($20$ nodes) and the layer 2 ($35$ nodes) of the Internet2 network~\cite{I2net}. (The results of the bounding technique are also given for reference.) Our simulation results show that the proposed technique, when compared to the MCS technique, provides approximations with higher accuracy, for less running time. 
\end{itemize}

The rest of the paper is organized as follows. Section~\ref{sec:PS} gives basic definitions/notations and the problem formulation. In Section~\ref{sec:Prel}, we overview the concepts of failure and repair rates, and cutsets and near-minimum cutsets. Sections~\ref{sec:App1} and~\ref{sec:App2} discuss the previous works and the proposed algorithms for $k$-terminal and all-terminal reliability systems, respectively. In Section~\ref{sec:SR}, we present our simulations results and compare the proposed technique for all-terminal reliability systems with the Monte Carlo simulation technique. Section~\ref{sec:COP} concludes the paper and discusses some open problems. The proofs of some lemmas are deferred to the appendix. 
	
\section{Problem Setup}\label{sec:PS}
Let $\Psi$ be a system with $n$ nodes, $k$ ($2\leq k\leq n$) of which are \emph{terminals} and the rest are \emph{relays}, and $m$ components (i.e., the edges connecting the nodes). The system $\Psi$ is called a \emph{$k$-terminal reliability system}, and for the case of $k=n$, the system $\Psi$ is called an \emph{all-terminal reliability system}. Fig.~\ref{fig:System} depicts a $3\times 3$ grid network as an example of an all-terminal reliability system with $n=k=9$ and $m=12$. (Note that the regularity of the system in Fig.~\ref{fig:System} is not a requirement for the applicability of the proposed techniques in this paper.)

Let $[m]\triangleq\{1,\dots,m\}$ be the index set of components in the system $\Psi$. Each component $i\in[m]$ is assumed to have two states: \emph{available} (up) and \emph{unavaiable} (down). We consider the setting that each component of the system $\Psi$ is subject to a statistically-independent continuous-time failure/repair random process described shortly, and study the steady-state behavior of the system $\Psi$ under such random processes.


\subsection{Failure/Repair Random Processes}
We assume that every component is initially available. Note that for a stationary ergodic stochastic process, such as a two-state process, the (steady-state) failure probability and (steady-state) failure frequency do not depend on the initial conditions~\cite{SB:77}. As time evolves, each component $i$ becomes unavailable after a random period of time, distributed \emph{arbitrarily} with mean $1/\lambda_i$ (for arbitrary $\lambda_i>0$), and it becomes available again after a random period of time, distributed \emph{arbitrarily} with mean $1/\mu_i$ (for arbitrary $\mu_i>0$). This process, for every component, continues over time, statistically independent from other components. 

Note that we do not restrict the distributions of up-times and down-times to be exponential. As shown in~\cite{S:80}, the (steady-state) failure probability and (steady-state) failure frequency of a system with two-state components depend only on the mean-up-times $\{1/\lambda_i\}$ and mean-down-times $\{1/\mu_i\}$ of the components in the system, and not on the up/down-time distribution functions per se.

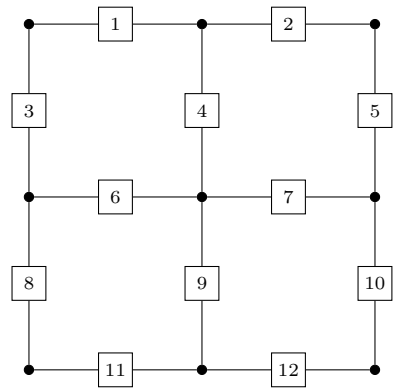
\begin{figure}
\centering
\begin{tikzpicture}
\tikzstyle{node}=[circle,draw=black,fill=black,minimum size=3.5pt,inner sep=0pt]
\tikzstyle{component}=[rectangle,draw=black,
                        inner sep=0pt,minimum size=4.5mm]
  \node at ( -3*1.15,4*1.15) [node] (n1) {};
  \node at ( -2*1.15,4*1.15) [component] (c1) {\scriptsize{$1$}};
  \node at ( -1*1.15,4*1.15) [node] (n2) {};
  \node at ( 0*1.15,4*1.15) [component] (c2) {\scriptsize{$2$}};
  \node at ( 1*1.15,4*1.15) [node] (n3) {};
  \node at ( -3*1.15,3*1.15) [component] (c4) {\scriptsize{$3$}};
  \node at ( -1*1.15,3*1.15) [component] (c5) {\scriptsize{$4$}};
  \node at ( 1*1.15,3*1.15) [component] (c6) {\scriptsize{$5$}};
  \node at ( -3*1.15,2*1.15) [node] (n5) {};
  \node at ( -2*1.15,2*1.15) [component] (c8) {\scriptsize{$6$}};
  \node at ( -1*1.15,2*1.15) [node] (n6) {};
  \node at ( 0*1.15,2*1.15) [component] (c9) {\scriptsize{$7$}};
  \node at ( 1*1.15,2*1.15) [node] (n7) {};
  \node at ( -3*1.15,1*1.15) [component] (c11) {\scriptsize{$8$}};
  \node at ( -1*1.15,1*1.15) [component] (c12) {\scriptsize{$9$}};
  \node at ( 1*1.15,1*1.15) [component] (c13) {\scriptsize{$10$}};
  \node at ( -3*1.15,0*1.15) [node] (n9) {};
  \node at ( -2*1.15,0*1.15) [component] (c15) {\scriptsize{$11$}};
  \node at ( -1*1.15,0*1.15) [node] (n10) {};
  \node at ( 0*1.15,0*1.15) [component] (c16) {\scriptsize{$12$}};
  \node at ( 1*1.15,0*1.15) [node] (n11) {};
  \draw (n1) -- (c1);
  \draw (c1) -- (n2);
  \draw (n2) -- (c2);
  \draw (c2) -- (n3);
  \draw (n1) -- (c4);
  \draw (c4) -- (n5);
  \draw (n2) -- (c5);
  \draw (c5) -- (n6);
  \draw (n3) -- (c6);
  \draw (c6) -- (n7);
  \draw (n5) -- (c8);
  \draw (c8) -- (n6);
  \draw (n6) -- (c9);
  \draw (c9) -- (n7);
  \draw (n5) -- (c11);
  \draw (c11) -- (n9);
  \draw (n6) -- (c12);
  \draw (c12) -- (n10);
  \draw (n7) -- (c13);
  \draw (c13) -- (n11);
  \draw (n10) -- (c16);
  \draw (c16) -- (n11);
  \draw (n9) -- (c15);
  \draw (c15) -- (n10);	
\end{tikzpicture}
\caption{An example of an all-terminal reliability system with $9$ nodes/terminals and $12$ components.}
\label{fig:System}
\end{figure}

\subsection{Failure Probability and Failure Frequency}
At any time, the (sub-) system of $\Psi$ including all $n$ nodes restricted only to available components (i.e., the original system excluding unavailable components), is referred to as the \emph{surviving system}. At any time, the system is \emph{unavailable} if the surviving system fails to connect all $k$ terminals. The (steady-state) \emph{probability of failure} is the probability that the system is unavailable, and the (steady-state) \emph{frequency of failure} is the expected number of times per unit time (i.e., the expected rate) that the system becomes unavailable~\cite{SB:77}.

\subsection{Problem Statement}
For arbitrary $\epsilon>0$ and $0<\delta<1$, we consider the problem of computing $(\epsilon,\delta)$-\emph{approximations} of $P_f$ and $F_f$ (with an error factor of at most $\epsilon$ and an error probability of at most $\delta$), denoted by $\tilde{P}_{f}$ and $\tilde{F}_f$, respectively, defined as \[\Pr\left\{|\tilde{P}_f-P_f|\geq \epsilon P_f \right\}\leq \delta,\] and \[\Pr\left\{|\tilde{F}_f-F_f|\geq \epsilon F_f \right\}\leq \delta.\] 

\subsection{Complexity Notation}\label{subsec:CN}
Throughout the paper, we follow the conventional use of notations in complexity theory as described below. 

For two arbitrary functions $f(x)$ and $g(x)$ of variable $x$, we write: (i) $f(x)=O(g(x))$ if for sufficiently large $x$, $|f(x)|\leq c\cdot |g(x)|$ for some constant $c>0$; (ii) $f(x)=\Omega(g(x))$ if for sufficiently large $x$, $f(x)\geq c\cdot g(x)$ for some constant $c>0$; (iii) $f(x)=\Theta(g(x))$ if $f(x)=O(g(x))$ and $f(x)=\Omega(g(x))$; and (iv) $f(x)=o(g(x))$ if for any constant $c>0$, $|f(x)|\leq c\cdot |g(x)|$ for sufficiently large $x$.

\section{Preliminaries}\label{sec:Prel}
\subsection{Failure/Repair Rates and Unavailability Probabilities}\label{subsec:FRR}

We refer to $\lambda_i$ and $\mu_i$ as the \emph{failure rate} and the \emph{repair rate} of component $i$, respectively. We assume that $\lambda_i$ and $\mu_i$ are independent of time (i.e., the failure/repair random process is stationary). We also assume that $\lambda_i$ and $\mu_i$ do not depend on $n$. Since $\lambda_i$ and $\mu_i$ are independent of time, the probability that each component $i$ is unavailable at any given time is equal to $p_i\triangleq {\lambda_i}/{(\lambda_i+\mu_i)}$. Note that the probabilities $\{p_i\}$ are independent of time. We refer to $p_i$ as the \emph{unavailability probability} of component $i$, and refer to $w_i=-\log p_i$ as the \emph{weight} of component $i$, where the symbol ``$\log$'', here and throughout the paper, refers to the natural logarithm. Since $0<p_i<1$ for all $i$, the weights $\{w_i\}$ are all non-negative. 


Define $\lambda_{\text{max}}\triangleq \max_{i\in [m]} \lambda_i$ as the maximum failure rate of a component. Similarly, define $\mu_{\text{min}}$ as the minimum repair rate of a component. For simplifying the arguments, we assume that $\mu_{\text{min}}/\lambda_{\text{max}}> m-1$. Define $w_{\text{max}}$ as the maximum weight of a component. Moreover, let $\lambda\triangleq \sum_{i=1}^{m} \lambda_i$, $\mu\triangleq\sum_{i=1}^{m} \mu_i$, and $w\triangleq \sum_{i=1}^{m} w_i$. 

We assume, without loss of generality, that there are no parallel components in the system. Otherwise, if there exist $l$ parallel components $\{i_j\}_{1\leq j\leq l}$ with failure rates $\{\lambda_{i_j}\}$, repair rates $\{\mu_{i_j}\}$, and unavailability probabilities $\{p_{i_j}\}$, we can replace them all by one component with unavailability probability $p=\prod_{j} p_{i_j}$, repair rate $\sum_{j} \mu_{i_j}$, and failure rate $(p\sum_{j}\mu_{i_j})/(1-p)$. This transformation does not change the failure probability and the failure frequency of a system~\cite{SB:77}. Note that this assumption is made only for the ease of exposition, and is not a requirement for the proposed algorithms. 




\subsection{Cutsets and Near-Minimum Cutsets}\label{subsec:CNMC}
A set of components in a system is a \emph{cutset} if the unavailability of those components yields the unavailability of the system. Moreover, a cutset is a \emph{minimal cutset} if it does not contain any other cutsets. Hereafter, we use the term ``cutset'' as a shorthand for ``minimal cutset.'' Let $N$ be the number of cutsets in a system, and let $C \triangleq \{\mathcal{C}_1,\dots,\mathcal{C}_N\}$ be the set of all cutsets in the system. Let $\mathcal{C}_j \triangleq \{{i_1},\dots,{i_{|\mathcal{C}_j|}}\}$, where ${i_1},\dots,{i_{|\mathcal{C}_j|}}$ are the \emph{indices of components} in $\mathcal{C}_j$, and $|\mathcal{C}_j|$ is the \emph{size} of $\mathcal{C}_j$, i.e., the number of components in $\mathcal{C}_j$. Hereafter, we assume that $N>1$ (otherwise, computing $P_f$ and $F_f$ is trivial). For example, in Fig.~\ref{fig:System}, there exist $4$ and $16$ cutsets of sizes $2$ and $3$, respectively, as enumerated in Table~\ref{tab:Cutsets}. 

For any $\mathcal{I}\subseteq [N]$, let $\mathcal{C}_{\mathcal{I}}\triangleq \cup_{j\in \mathcal{I}} \mathcal{C}_j$. The probability of failure of cutsets $\{\mathcal{C}_j\}_{j\in\mathcal{I}}$, denoted by $p(\mathcal{C}_{\mathcal{I}})$, is equal to the probability that all components in $\mathcal{C}_j$, for all ${j\in \mathcal{I}}$, are unavailable, i.e., $p(\mathcal{C}_{\mathcal{I}}) \triangleq \prod_{i\in \mathcal{C}_{\mathcal{I}}} p_{i}$. Note that $p(\mathcal{C}_j)$ is the probability of unavailability of all components in $\mathcal{C}_j$. Let $w(\mathcal{C}_j)$, the \emph{weight} of $\mathcal{C}_j$, be the sum of the weights of all components in $\mathcal{C}_j$, i.e., $w(\mathcal{C}_j) \triangleq \sum_{i\in \mathcal{C}_j} w_{i}$. Note that $w(\mathcal{C}_j) = -\log(p(\mathcal{C}_j))$. Let $w^{*}\triangleq\min_{\mathcal{C}\in C} w(\mathcal{C})$ and $p^{*}\triangleq\max_{\mathcal{C}\in C} p(\mathcal{C})$ be the minimum weight and maximum failure probability of a cutset, respectively. Let $s^{*}\triangleq\min_{\mathcal{C}\in  C} |\mathcal{C}|$ be the minimum size of a cutset. Note that $s^{*}\geq \min\{\max\{w^{*}/w_{\text{max}},1\},m-1\}$. 



For any constant $\alpha\geq 1$ (independent of $n$ and $m$), let $C^{(\alpha)}\triangleq\{\mathcal{C}\in C: w(\mathcal{C})\leq \alpha w^{*}\}$ be the set of all (minimal) cutsets in a system of weight less than or equal to $\alpha w^{*}$. Let $N^{(\alpha)}$ be the number of cutsets in $C^{(\alpha)}$. We refer to the cutsets in $C^{(\alpha)}$ as \emph{$\alpha$-min cutsets}. For simplicity, we refer to $1$-min cutsets as \emph{min-cutsets}. Note that $w(\mathcal{C})=w^{*}$ and $p(\mathcal{C})=p^{*}$ for all min-cutsets $\mathcal{C}\in C$.

For the cases in which all components have the same weight $w_0$, there is a one-to-one map between the weights and the sizes of the cutsets (i.e., the weight of each cutset is equal to the size of that cutset times $w_0$). For example, for the case that all components in Fig.~\ref{fig:System} have unit weight (i.e., all cutsets of size $2$ or size $3$ have weight $2$ or weight $3$, respectively), Table~\ref{tab:Cutsets} enumerates the $1.5$-min cutsets of the system in Fig.~\ref{fig:System}.


\begin{table}[t]
\centering
\caption{Minimal Cutsets of Size $2$ and $3$ in the System of Fig.~\ref{fig:System}}
\label{tab:Cutsets}
\begin{tabular}{|c|cccc|}
\hline
Cutsets of Size $2$  & \multicolumn{4}{c|}{Cutsets of Size $3$} \\ \hline
$\{1,3\}$ & $\{1,2,4\}$ & $\{2,3,4\}$ & $\{1,6,11\}$ & $\{8,9,10\}$ \\
$\{2,5\}$ & $\{1,6,8\}$ & $\{3,4,5\}$ & $\{2,7,10\}$ & $\{8,9,12\}$ \\
$\{8,11\}$ & $\{1,4,5\}$ & $\{5,7,10\}$ & $\{2,7,12\}$ & $\{9,10,11\}$ \\
$\{10,12\}$ & $\{3,6,8\}$ & $\{5,7,12\}$ & $\{3,6,11\}$ & $\{9,11,12\}$ \\ \hline
\end{tabular}
\end{table}



\section{A Poly$(N)$-Time Approximation Algorithm for $k$-Terminal Reliability Systems}\label{sec:App1}












\subsection{Background}\label{subsec:PreWork1}
In this section, we overview the inclusion-exclusion based formulas for $P_f$ and $F_f$, and the bounding technique for approximating $P_f$ and $F_f$. The reader familiar with these concepts can skip this section. 

By the cutset approach~\cite{SB:77} based on the inclusion-exclusion principle, $P_f$ and $F_f$ can be written as:
\begin{dmath}\label{eq:Pf}
P_f = \sum_{l=1}^{N} \left((-1)^{l+1} \sum_{\mathcal{I}\subseteq [N]: |\mathcal{I}|=l} p(\mathcal{C}_{\mathcal{I}})\right)
\end{dmath} and 
\begin{dmath}\label{eq:Ff}
F_f = \sum_{l=1}^{N} \left((-1)^{l+1} \sum_{\mathcal{I}\subseteq [N]: |\mathcal{I}|=l} \left(p(\mathcal{C}_{\mathcal{I}})\sum_{i\in \mathcal{C}_{\mathcal{I}}}\mu_i\right)\right)
\end{dmath} (The details of derivation of the formulas~\eqref{eq:Pf} and~\eqref{eq:Ff} can be found in~\cite{Sn:79}.) Note that, depending on the topology (e.g., ring topology), the number of cutsets ($N$) in some systems is only poly$(n)$, and all cutsets can be enumerated in $O(N)$ time~\cite{TSOA:80}. However, the complexity of computing $P_f$ and $F_f$ using the formulas \eqref{eq:Pf} and \eqref{eq:Ff} may be still unaffordable. (The number of terms in formulas~\eqref{eq:Pf} and~\eqref{eq:Ff} is exponential in the number of cutsets ($N$) and double-exponential in the number of nodes ($n$). Thus, there does not exist a poly$(n)$-time (i.e., with running time polynomial in $n$) nor a poly$(N)$-time (i.e., with running time polynomial in $N$) algorithm for computing $P_f$ and $F_f$ from~\eqref{eq:Pf} and~\eqref{eq:Ff}, directly.) For such systems, one may use the bounding technique to approximate $P_f$ or $F_f$. 

The bounding technique is one of the most common approaches to provide upper and lower bounds on $P_f$ or $F_f$ via truncating the formula~\eqref{eq:Pf} or~\eqref{eq:Ff}, respectively~\cite{S:77}. 


Let 
\begin{dmath*}
P_f^{+} = \sum_{1\leq j\leq N} p(\mathcal{C}_j),
\end{dmath*} and 
\begin{dmath*}
P_f^{-} = \sum_{1\leq j\leq N} p(\mathcal{C}_j)-\sum_{1\leq j_1<j_2\leq N} p(\mathcal{C}_{j_1}\cup \mathcal{C}_{j_2}).	
\end{dmath*} Similarly, let 
\begin{dmath*}
F_f^{+} =\sum_{1\leq j\leq N} \left(p(\mathcal{C}_j) \sum_{i\in \mathcal{C}_j} \mu_i\right),
\end{dmath*} and 
\begin{dmath*} F_f^{-} = \sum_{1\leq j\leq N} \left(p(\mathcal{C}_j) \sum_{i\in \mathcal{C}_j} \mu_i\right) - \sum_{1\leq j_1<j_2\leq N}\left(p(\mathcal{C}_{j_1} \cup \mathcal{C}_{j_2})\sum_{i\in \mathcal{C}_{j_1}\cup \mathcal{C}_{j_2}}\mu_i \right).	
\end{dmath*}
 Then, $P_f^{+}$ and $P_f^{-}$ (or $F_f^{+}$ and $F_f^{-}$) are the \emph{first-order upper-bound} and \emph{lower-bound} on $P_f$ (or $F_f$), respectively. Using a similar technique by incorporating larger collections of cutsets (instead of singletons or pairs only), one can compute higher-order upper- and lower-bounds on $P_f$ and $F_f$ \cite{S:77}. Such higher-order bounds, when compared to the first-order bounds, are more accurate, but more computationally expensive.

Now a question is what type of guarantee the bounding technique provides on the accuracy of the approximation. To answer this question, let $d_P$ (or $d_F$) be the maximum number of decimal places up to which $P_f^{+}$ and $P_f^{-}$ (or $F_f^{+}$ and $F_f^{-}$) match. Let $\hat{P}_f = \mathrm{trunc}(P_f^{+},d_P) = \mathrm{trunc}(P_f^{-},d_P)$ and $\hat{F}_f=\mathrm{trunc}(F_f^{+},d_F)=\mathrm{trunc}(F_f^{-},d_F)$, where $\mathrm{trunc}(x,d)={\lfloor 10^d\cdot x\rfloor}/{10^d}$, for any real number $x\geq 0$ and integer $d\geq 1$. Note that $\hat{P}_f$ and $\hat{F}_f$ are the most accurate estimators of $P_f$ and $F_f$ based on the (first-order) upper and lower bounds. Both bounds $\hat{P}_f$ and $\hat{F}_f$ are computable in poly$(N)$ time, yet neither guarantees an arbitrary approximation error factor. To be more specific, $\hat{P}_f$ and $\hat{F}_f$ are always exact up to $d_P$ and $d_F$ decimal places, respectively, but the approximation error depends on $d_P$ and $d_F$, and cannot be made arbitrarily small.


\subsection{Main Ideas of the Proposed Algorithm}\label{subsec:MainIdea1}
In this section, we give an overview of the KLM estimator as part of the proposed algorithms, and explain our main ideas. 

Let the symbols ``$\wedge$'' and ``$\vee$'' denote the \emph{logical conjunction} (AND) and the \emph{logical disjunction} (OR), respectively. Let \[\Phi_{M}=Z_1\vee Z_2\vee\dots\vee Z_{M}\] be a formula on $M>1$ variables $\{Z_j\}_{j=1}^{M}$, where the \emph{clause} $Z_j$ is a conjunction of \emph{literals} $z_{i}$ for some $i\in [m]$, i.e., $Z_j=\wedge_{i\in I_j} z_{i}$ for some $I_j\subseteq [m]$. Each literal $z_{i}$ is either a Boolean variable or the negation of a Boolean variable, and it takes two values: ``true'' and ``false.'' The formula $\Phi_{M}$ of such form is said to have \emph{disjunctive normal form} (DNF). Let $P_Z(j)$ be the probability that the clause $Z_j$ is true (i.e., the literals $\{z_{i}\}_{i\in I_j}$ are all true), and let $P_Z$ denote the vector $[P_Z(1),\dots,P_Z(M)]$. Note that the truth probability of $\Phi_M$, denoted by $\Pi_{M}$, cannot be computed in poly$(M)$-time (i.e., polynomial-time in $M$) \cite{KLM:89}. However, using an unbiased estimator, referred to as KLM, due to Karp, Luby, and Madras \cite{KLM:89}, one can compute a $(\xi,\delta)$-approximation of $\Pi_{M}$ in poly$(M)$ time (for any given $\xi>0$ and $0<\delta<1$). In the following, we describe a simple, yet powerful, extension of the KLM estimator. 

The KLM estimator with inputs $(\Phi_{M},P_Z; \xi,\delta)$ proceeds in steps as follows: 
\begin{itemize}
\item[0.] Initialize the counters $s$ and $t$ by setting $s=t=1$;
\item[1.] Choose a random clause $Z_j$, with probability of selecting $Z_j$ being equal to ${P_Z(j)}/{Q_Z}$, where $Q_Z=\sum_{j} P_Z(j)$;
\item[2.] Choose a random assignment $\boldsymbol{z}=\{z_1,\dots,z_m\}$ satisfying the clause $Z_j$ (i.e., $z_i$ is true for all $i\in I_j$); 
\item[3.] Compute $\pi_{s,t}={Q_Z}/{N(\boldsymbol{z})}$, where $N(\boldsymbol{z})$ is the number of clauses that the assignment $\boldsymbol{z}$ satisfies; 
\item[4.] $s\leftarrow s+1$  
\item[5.] Repeat Steps 1-4 $S=\lceil {4(M-1)}/{\xi^2}\rceil$ times;
\item[6.] Compute the mean $\pi_{t}=\left({\sum_{s=1}^{S} \pi_{s,t}}\right)/{S}$; 
\item[7.] $t\leftarrow t+1$ 
\item[8.] Repeat Steps 1-7 $T = \lceil 12 \log({1}/{\delta})\rceil$ times;
\item[9.] Return the median of $\{\pi_{t}\}_{t=1}^{T}$.
\end{itemize} 

The running time of an obvious implementation of the KLM estimator is $O((M^2 m/\xi^2) \log(1/\delta))$. In particular, Step~1 takes $O(M)$ time to run. Step~2 can be simply run in $O(m)$ time; Step~3 takes $O(M m)$ time to run; and each of Steps~1-3 is run $ST = O((M/\xi^2)\log(1/\delta))$ times; Step~6 can be run in $O(S)$ time, and this step is run $T$ times, and Step~9 can be run in $O(T)$ time, and this step is run only once. Note that a more sophisticated implementation of the KLM estimator, referred to as \emph{self-adjusting}, can be run in $O((M m/\xi^2)\log(1/\delta))$ time (see, for more details, \cite{KLM:89}).

Let $\tilde{\Pi}_{M}$ be the output of the KLM estimator. Then, the following result holds.

\begin{lemma}\label{lem:KLME}
$\tilde{\Pi}_{M}$ is a $({\xi},{\delta})$-approximation of $\Pi_M$.\end{lemma}

\begin{proof}
The proof follows from similar arguments as those in~\cite{KLM:89}, and can be found in the appendix.	
\end{proof}

As can be seen in \eqref{eq:Pf}, $P_f$ is the probability of union of a set of events, and consequently, it can be thought of as the probability of satisfying a DNF formula with random Boolean variables. Thus, the KLM estimator can compute an $(\epsilon,\delta)$-approximation of $P_f$ for $k$-terminal reliability systems (for any $k$) in poly$(N)$ time. However, as one can see in~\eqref{eq:Ff}, $F_f$ is not the probability of union of any set of events. Thus, $F_f$ cannot be thought of as the probability of satisfying a DNF formula, and the KLM estimator is not directly applicable. This poses a challenge for approximating $F_f$. To overcome this challenge, we define an auxiliary term $P$ that satisfies the following requirements: (i) $P$ can be interpreted as the probability of the union of an auxiliary set of events, and (ii) $F_f$ is a scalar multiple of the difference between $P_f$ and $P$. Note that, for the first time in the literature, this work presents a linear connection between $P_f$ and $F_f$. 

Each of $P_f$ and $P$ is the probability of union of a set of events, and thus can be approximated by the KLM estimator within an arbitrary multiplicative error. Now the main questions are whether $P_f$ and $P$ can be approximated in polynomial time in the number of cutsets, and whether combining such approximations of $P_f$ and $P$, a multiplicative approximation of $F_f$ with an arbitrary error factor can be computed. In this section, we answer these questions in the affirmative, and propose a poly$(N)$-time algorithm that computes an $(\epsilon,\delta)$-approximation of $F_f$ for $k$-terminal reliability systems (for any $k$).  

We define $P$ that satisfies the requirements (i) and (ii) as follows:

\begin{dmath}\label{eq:P}\nonumber
P = \sum_{l=1}^{N} \left((-1)^{l+1}\sum_{\mathcal{I}\subseteq [N]: |\mathcal{I}|=l} \left(P(\mathcal{C}_{\mathcal{I}}) \left(1-\frac{\sum_{i\in \mathcal{C}_{\mathcal{I}}}\mu_i}{\mu}\right)\right)\right).
\end{dmath} By using~\eqref{eq:Pf}--\eqref{eq:P}, it is easy to verify that 
\begin{equation}\label{eq:FfPfPmmu}
F_f=(P_f-P)\mu,	
\end{equation}
where $F_f$ and $P_f$ are given by~\eqref{eq:Pf} and~\eqref{eq:Ff}, respectively. Thus, $F_f$ is a scalar multiple of the difference between $P_f$ and $P$, as desired. Note that $P_f> P$ since $F_f>0$. (We notice, without proof, that $F_f=0$ only in a system with $p_i=0$ for all $i$, and this contradicts with the assumption that $0<p_i<1$.)

Now, we define an auxiliary set of events such that $P$, defined in~\eqref{eq:P}, is the probability of union of these events. In the steady state, at any given time, each component $i$ is unavailable or available, with probability $p_i$ or $1-p_i$, respectively, independent of time \cite{SB:77}. Thus, the random process under consideration (Section~\ref{sec:PS}) is equivalent to the following \emph{one-shot} random process over the system $\Psi$. Each component $i$, statistically independently from other components, is set to be ``unavailable'' with probability $p_i$, and it is set to be ``available'' otherwise. We refer to this process as \emph{sampling}. 

We further introduce an auxiliary \emph{one-shot} random process over the system $\Psi$ as follows. Each component is assumed to have two states: \emph{exposed} and \emph{unexposed}. One component, say $i$, is chosen with probability $\mu_i/\mu$, and is set to be exposed; and the rest of the components are set to be unexposed. We refer to this process as \emph{exposure}. Note that the sampling and exposure processes are statistically independent. 

The intuition behind the sampling and exposure processes is as follows. Each term in $P$ corresponds to a collection of cutsets and expresses the probability that all components in this collection of cutsets are unavailable and unexposed. Moreover, as can be seen in~\eqref{eq:P}, $P$ has the structure of an inclusion-exclusion formula. Thus, it should not be hard to see that $P$ is equal to the probability that all components in some cutset of the system $\Psi$ (under the sampling and exposure processes) are unavailable and unexposed (due to the statistical independence of these processes). 


\subsection{Proposed Algorithm}\label{subsec:PA1}
For a $k$-terminal reliability system (for arbitrary $k$), the inputs of the proposed algorithm are the failure rates $\{\lambda_i\}$, the repair rates $\{\mu_i\}$, and the approximation parameters $\epsilon>0$ and $0<\delta<1$. The algorithm proceeds in steps as follows:
\begin{itemize}
\item[0.] Initialization:
\begin{itemize}
\item[0.1] Enumerate all $N$ (minimal) cutsets of the system;
\item[0.2] Compute $s^{*}$, $\mu_{\text{min}}$, $\lambda_{\text{max}}$, $\mu$, and $\rho=\mu_{\text{min}}s^{*}-\lambda_{\text{max}} (m-s^{*})$;
\item[0.3] Take $\xi=(\epsilon/2)(\rho/\mu)$
\end{itemize}
\item[1.] Compute a $(\xi,\delta/2)$-approximation $\tilde{P}$ of $P$ using the KLM estimator; 
\item[2.] Compute a $(\xi,\delta/2)$-approximation $\tilde{P}_f$ of $P_f$ using the KLM estimator; 
\item[3.] Return $\tilde{F}_f = (\tilde{P}_f-\tilde{P})\mu$. 	
\end{itemize}


The running time of this algorithm is $O((N m^3\log N) (1/\epsilon^2)\log(1/\delta))$. (The running time of each step of the algorithm is given in Section~\ref{subsec:CC1}.)

The details of the computations of $\tilde{P}_f$ and $\tilde{P}$ (Steps~1 and~2) are as follows. 

\subsubsection{Computation of $\tilde{P}$}\label{subsubsec:TildeP1}
Construct a DNF formula as \[\Phi_{N}\triangleq Z_1\vee Z_2\vee\dots\vee Z_{N},\] where the clause $Z_j$, $j\in [N]$, is the conjunction of two \emph{sub-clauses} $X_j$ and $Y_j$ (i.e., $Z_j=X_j \wedge Y_j$). Define \[X_j\triangleq x_{i_1}\wedge x_{i_2}\wedge\dots\wedge x_{i_{|\mathcal{C}_j|}},\] and \[Y_j\triangleq y_{i_1}\wedge y_{i_2}\wedge\dots\wedge y_{i_{|\mathcal{C}_j|}},\] where $i_1,\dots,i_{|\mathcal{C}_j|}$ are the labels of the components in the cutset $\mathcal{C}_j$, and the literals $x$ and $y$ are Boolean (random) variables defined as follows. For any (random) assignment of literals, $x_i$ is true with probability $p_i=\lambda_i/(\lambda_i+\mu_i)$, and it is false, otherwise; and $y_i$ is true for all $i$, except for one and only one $i$ being chosen with probability $\mu_{i}/\mu$. Define \[P_Z(j) \triangleq p(\mathcal{C}_j) \left(1-\frac{\sum_{i\in \mathcal{C}_j} \mu_i}{\mu}\right).\] Run the KLM estimator with inputs $(\Phi_N,P_Z;\xi,\delta/2)$, and denote by $\tilde{P}$ the output.

\subsubsection{Computation of $\tilde{P}_f$}\label{subsubsec:TildePf1}
The method of computing $\tilde{P}_f$ is similar to that of $\tilde{P}$, except that a slightly different DNF formula is required. This technique was previously used in~\cite{K:01}. Construct a DNF formula as \[\Phi_{N,f}\triangleq X_1\vee X_2\vee\dots\vee X_{N},\] where the clause $X_j$, $j\in [N]$, is defined as before. Define \[P_{X}(j)\triangleq p(\mathcal{C}_j).\] Run the KLM estimator with inputs $(\Phi_{N,f},P_X;\xi,\delta/2)$, and denote by $\tilde{P}_f$ the output. 

\subsection{Theoretical Analysis}\label{subsec:TA1}

\begin{theorem}\label{thm:Main1}
The output of the algorithm in Section~\ref{subsec:PA1}, $\tilde{F}_f$, is an $(\epsilon,\delta)$-approximation of $F_f$. 
\end{theorem}

\begin{proof} By the definitions of $X_j$ and $Y_j$ as above, it is not hard to see that (i) $P_{X}(j)$ is the probability that a random assignment $\boldsymbol{x}=\{x_{1},\dots,x_m\}$ satisfies $X_j$ (i.e., $X_j$ is true), and (ii) $P_Y(j)\triangleq 1-\sum_{i\in \mathcal{C}_j}\mu_i/\mu$, is the probability that a random assignment $\boldsymbol{y}=\{y_{1},\dots,y_m\}$ satisfies $Y_j$ (i.e., $Y_j$ is true). By (i) and (ii), it follows that a random assignment $\boldsymbol{z}=(\boldsymbol{x},\boldsymbol{y})$ satisfies $Z_j$ (i.e., $Z_j$ is true) with probability $P_Z(j)$, defined earlier, since $X_j$ and $Y_j$ are statistically independent, i.e., $P_Z(j) = P_X(j)\cdot P_Y(j)$.  

The formula $\Phi_N$ is true so long as some clause $Z$ is true, and thus by the inclusion-exclusion principle, it immediately follows that the truth probability $\Pi_N$ of $\Phi_{N}$ is equal to $P$. Thus, in order to approximate $P$, it suffices to approximate $\Pi_N$. Since $\Phi_N$ is a DNF formula, $\Pi_N$ can be approximated using the KLM estimator. By Lemma~\ref{lem:KLME}, the output of the KLM estimator with inputs $(\Phi_N,P_Z;\xi,\delta/2)$, denoted by $\tilde{P}$, is a $(\xi,\delta/2)$-approximation of $\Pi_N$, or equivalently, $P$. This yields the following result. 

\begin{lemma}\label{lem:Ptilde}
$\tilde{P}$ is a $(\xi,\delta/2)$-approximation of $P$.
\end{lemma}

Moreover, the DNF formula $\Phi_{N,f}$ is true so long as $X_j$ is true for some $j$. Thus, the probability that $\Phi_{N,f}$ is true, denoted by $\Pi_{N,f}$, is equal to $P_f$. By a similar argument as above, the following result is immediate.

\begin{lemma}\label{lem:Pftilde}
$\tilde{P}_f$ is a $(\xi,\delta/2)$-approximation of $P_f$.
\end{lemma}

The results of Lemma~\ref{lem:Ptilde} and Lemma~\ref{lem:Pftilde} yield 
\begin{equation}\label{eq:TildePNoAlphaM1stBound}
\Pr\left\{|\tilde{P}-P|\geq \xi P\right\}\leq \frac{\delta}{2},
\end{equation} and
\begin{equation}\label{eq:TildePfNoAlphaM1stBound}
\Pr\left\{|\tilde{P}_f-P_f|\geq \xi P_f \right\}\leq \frac{\delta}{2}.
\end{equation} By combining~\eqref{eq:TildePNoAlphaM1stBound} and~\eqref{eq:TildePfNoAlphaM1stBound}, we get 
\begin{equation*}\label{eq:PPfNoAlphaBound}
\Pr\left\{|(\tilde{P}_f-\tilde{P})-\left(P_f-P\right)|\geq \xi (P_f+P) \right\}\leq \delta.
\end{equation*} Since $F_f=(P_f-P)\mu$ (by~\eqref{eq:FfPfPmmu}) and $\tilde{F}_f = (\tilde{P}_f-\tilde{P})\mu$ (by the algorithm), we get 
\begin{equation*}\label{eq:FfNoAlphaBound1}
\Pr\left\{|\tilde{F}_f-F_f|\geq \xi (P_f+P)\mu \right\}\leq \delta,
\end{equation*} or equivalently, 
\begin{equation}\label{eq:FfNoAlphaBound2}
\Pr\left\{|\tilde{F}_f-F_f|\geq \epsilon F_f \right\}\leq \delta,
\end{equation} so long as 
\begin{equation}\label{eq:EpsilonStar1}
\epsilon \geq \xi\left(\frac{P_f+P}{F_f}\right)\mu.
\end{equation} Thus, it suffices to show that our choice of $\xi$ meets the condition~\eqref{eq:EpsilonStar1}. To do so, we need to establish lower- and  upper bounds on $F_f$. This can be done by using the system-state based formulas for $P_f$ and $F_f$~\cite{MS:99}. 

Let $\boldsymbol{s}\triangleq\{s_1,\dots,s_m\}$ be the state of a system at a given time, where $s_i$ is ``true'' if the component $i$ is available at that time, and $s_i$ is ``false'' otherwise. Let $\mathcal{S}_f$ be the set of all states $\boldsymbol{s}$ in which the system is unavailable. Moreover, let $p(\boldsymbol{s})$ be the probability of the state $\boldsymbol{s}$, i.e., 
\begin{equation}\label{eq:psform}
p(\boldsymbol{s}) \triangleq \prod_{i\in \mathcal{I}(\boldsymbol{s})} p_i \cdot \prod_{i\in [m]\setminus \mathcal{I}(\boldsymbol{s})} (1-p_i),
\end{equation} where $\mathcal{I}(\boldsymbol{s})$ and $[m]\setminus \mathcal{I}(\boldsymbol{s})$ are the set of unavailable and available components, respectively, in the state $\boldsymbol{s}$. It was shown in~\cite{MS:99} that $P_f$ and $F_f$ can be written as: 
\begin{equation}\label{eq:PFMC}
P_f = \sum_{\boldsymbol{s}\in \mathcal{S}_f} p(\boldsymbol{s}),
\end{equation} and 
\begin{dmath}\label{eq:FFMC}
F_f = \sum_{\boldsymbol{s}\in \mathcal{S}_f} p(\boldsymbol{s})\left( \sum_{i\in \mathcal{I}(\boldsymbol{s})}\mu_i -\sum_{i\in [m]\setminus \mathcal{I}(\boldsymbol{s})}\lambda_i \right).
\end{dmath} 



Note that $|\mathcal{I}(\boldsymbol{s})|\geq s^{*}$ for any $\boldsymbol{s}\in \mathcal{S}_f$ since (i) all components of at least one cutset are unavailable, and (ii) the size of any cutset is bounded from below by $s^{*}$; and $|\mathcal{I}(\boldsymbol{s})|\leq m$ for all $\boldsymbol{s}$. Thus, 
\begin{equation}\label{eq:MuIs}
\mu_{\text{min}}s^{*}\leq \sum_{i\in \mathcal{I}(\boldsymbol{s})}\mu_i\leq \mu
\end{equation} and 
\begin{equation}\label{eq:LambdaIs} 
0\leq\sum_{i\in [m]\setminus \mathcal{I}(\boldsymbol{s})} \lambda_i\leq (m-s^{*})\lambda_{\text{max}}	
\end{equation} for all $\boldsymbol{s}\in \mathcal{S}_f$. By the choice of the algorithm, $\rho = \mu_{\text{min}}s^{*}-\lambda_{\text{max}} (m-s^{*})$ (Step~0.2). Note that $\rho = O(m)$ and $\rho=\Omega(1)$ since $1\leq s^{*}\leq m$, and $\mu_{\text{min}}=\Theta(1)$ and $\lambda_{\text{max}}=\Theta(1)$. Putting~\eqref{eq:MuIs} and~\eqref{eq:LambdaIs} together, it follows that $\sum_{i\in \mathcal{I}(\boldsymbol{s})}\mu_i-\sum_{i\in [m]\setminus \mathcal{I}(\boldsymbol{s})} \lambda_i$, for all $\boldsymbol{s}\in \mathcal{S}_f$, is lower and upper bounded by $\rho$ and $\mu$, respectively. By~\eqref{eq:PFMC} and~\eqref{eq:FFMC}, it is then easy to see that \begin{equation}\label{eq:PFFFPF}
P_f\rho\leq F_f\leq P_f \mu.
\end{equation}



Using the bounds on $F_f$ in~\eqref{eq:PFFFPF}, we proceed with the rest of the proof as follows. 

Since $F_f\geq P_f\rho$ (by~\eqref{eq:PFFFPF}), and $P_f+P< 2P_f$ (by the fact that $F_f>0$, and so, $P_f>P$ (by~\eqref{eq:FfPfPmmu})), one can see that
\begin{equation}\label{eq:TB1}
\xi\left(\frac{P_f+P}{F_f}\right)\mu< 2\xi\left(\frac{\mu}{\rho}\right).	
\end{equation} Taking $\xi=(\epsilon/2) (\rho/\mu)$ as in the algorithm, it is obvious that~\eqref{eq:EpsilonStar1}, and consequently,~\eqref{eq:FfNoAlphaBound2} hold (by~\eqref{eq:TB1}). Thus, $\tilde{F}_f$ is an $(\epsilon,\delta)$-approximation of $F_f$.\qed
\end{proof}

Note that the smaller is $\xi$, the larger is the running time of the proposed algorithm. However, the lower bound on $F_f$ and the upper bound on $P_f+P$ dictate the choice of $\xi$. Thus, the closer are these bounds to the actual values, the more efficient is the computation of the approximation. However, the bounds in~\eqref{eq:TB1} follow from a worst-case analysis. For improving on these bounds, the trick is to run the proposed algorithm multiple times, and amplify $\xi$ in each run based on the results of the previous runs. Similar idea was previously used in~\cite{KLM:89}. This is, however, beyond the scope of this paper, and hence not discussed here. 


\subsection{Computational Complexity}\label{subsec:CC1}
The initialization (Step~0) can be run in $O(N\log N)$ time. In Step~0.1, as was shown in~\cite{TSOA:80}, all the $N$ minimal cutsets can be enumerated in $O(N)$ time. In Step~0.2, $s^{*}$ can be computed in $O(N\log N)$ time (via sorting), $\mu_{\text{min}}$ and $\lambda_{\text{max}}$ in $O(m\log m)$ time (via sorting), and $\mu$ and $\rho$ in $O(m)$ time and $O(1)$ time, respectively.

By the choice of $\xi = (\epsilon/2)(\rho/\mu)$ in the algorithm, $\xi=\Omega(\epsilon/m)$ since $\rho = \Omega(1)$ and $\mu = O(m)$. In Steps~1 and~2, the estimates $\tilde{P}$ and $\tilde{P}_f$ can be computed in poly$(N)$ time by running the (self-adjusting) KLM estimator in $O((N m/\xi^2)\log(1/\delta))$ time~\cite{KLM:89}, or equivalently, $O((N m^3) (1/\epsilon^2)\log(1/\delta))$ time. Putting everything together, the proposed algorithm runs in poly$(N)$ time. This technique is useful for systems with poly$(n)$ cutsets. In such systems, $F_f$ can be approximated accurately and efficiently by the proposed algorithm; whereas the accuracy of approximating $F_f$ using the bounding technique may be inadequate (see Section~\ref{subsec:PreWork1}).

\section{A Poly$(n)$-Time Approximation Algorithm for All-Terminal Reliability Systems}\label{sec:App2}

\subsection{Background}\label{subsec:PreWork2}
In this section, we overview the Monte Carlo simulation for approximating $P_f$ and $F_f$ as part of the proposed algorithm. The reader familiar with this concept can skip this section. 


In many systems, there exist exponentially many cutsets and failure states, and it is not practical to identify and enumerate all cutsets or failure states. Thus $P_f$ and $F_f$ cannot be computed from~\eqref{eq:Pf} and~\eqref{eq:Ff} or~\eqref{eq:PFMC} and~\eqref{eq:FFMC}, directly. The Monte Carlo simulation (MCS) is useful to approximate $P_f$ and $F_f$ in such cases~\cite{MS:99}. This technique performs MCS over the system state space (see, for details,~\eqref{eq:psform}--\eqref{eq:FFMC}). 

The inputs of the MCS algorithm are the failure rates $\{\lambda_i\}$, the repair rates $\{\mu_i\}$, and the number of simulation runs, $S$ and $T$. The MCS algorithm proceeds in steps as follows (see Section~\ref{subsec:TA1} for the definitions of $\boldsymbol{s}$, $p(\boldsymbol{s})$, and $\mathcal{I}(\boldsymbol{s})$): 

\begin{itemize}
\item[0.] Initialize the counters $s$ and $t$ by setting $s=t=1$;
\item[1.] Choose a random state $\boldsymbol{s}$, with probability of selecting $\boldsymbol{s}$ being equal to $p(\boldsymbol{s})$;
\item[2.] If the system is unavailable in the state $\boldsymbol{s}$, let $\pi_{s,t} = 1$ and $\varphi_{s,t}=\sum_{i\in \mathcal{I}(\boldsymbol{s})}\mu_i-\sum_{i\in [m]\setminus \mathcal{I}(\boldsymbol{s})} \lambda_i$; otherwise, let $\pi_{s,t} =0$ and $\varphi_{s,t}=0$; 
\item[3.] $s\leftarrow s+1$  
\item[4.] Repeat Steps 1-3 $S$ times;
\item[5.] Compute the means $\pi_t=(\sum_{s=1}^{S} \pi_{s,t})/S$ and $\varphi_t=(\sum_{s=1}^{S} \varphi_{s,t})/S$; 
\item[6.] $t\leftarrow t+1$
\item[7.] Repeat Steps 1-6 $T$ times;
\item[8.] Return the median of $\{\pi_t\}_{t=1}^{T}$ and the median of $\{\varphi_t\}_{t=1}^{T}$. 
\end{itemize} 

The running time of MCS is $O(ST(m+n))$. In particular, Step 1 can be run in $O(m)$ time; Step 2 can be run in $O(m+n)$ time, e.g., using the breadth first search or the depth first search. Each of Steps~1 and~2 is run $ST$ times; Step~5 can be run in $O(S)$ time, and this step is run $T$ times; Step~6 can be run in $O(T)$ time, and this step is run only once. 

We will show in Section~\ref{subsec:TA2} that MCS can provide approximations of $P_f$ and $F_f$ with arbitrary additive error factors in poly$(n)$ time; whereas for approximating $P_f$ and $F_f$ within an arbitrary multiplicative error, MCS cannot be run in poly$(n)$ time.

\subsection{Main Ideas of the Proposed Algorithm}\label{subsec:MainIdeas2}
The proposed algorithm in Section~\ref{sec:App1} provides an approximation of $P_f$ and $F_f$ in poly$(N)$-time with provable guarantees for any $k$-terminal reliability system. However, the number of cutsets ($N$) generally grows exponentially with the number of nodes ($n$). Thus, a natural question that arises is whether one can design an algorithm for computing $(\epsilon,\delta)$-approximations of $P_f$ and $F_f$ which runs in poly$(n)$-time. Karger in~\cite{K:01} proposed the first (and only) poly$(n)$-time algorithm for approximating $P_f$ for all-terminal reliability systems (i.e., $k=n$). In this section, we present the first poly$(n)$-time algorithm for approximating $F_f$ for all-terminal reliability systems. 

As was previously shown in \cite{KS:96}, in all-terminal reliability systems, the number of minimum cutsets, i.e., those cutsets with minimum weight, is poly$(n)$. Moreover, there are a poly$(n)$ number of near-minimum cutsets with weight not greater than a given constant factor of the weight of the minimum cutsets (see Lemma~\ref{lem:AlphaMinCutsNumber}), and such cutsets can all be enumerated in poly$(n)$ time (see Lemma~\ref{lem:AlphaMinCutsetsEnum}). Note that, for a poly$(n)$ number of cutsets, there still exist an exponential number of terms in \eqref{eq:Pf} and \eqref{eq:Ff}, and consequently, using \eqref{eq:Pf} and \eqref{eq:Ff} one can only provide a series of upper- and lower-bounds on $P_f$ and $F_f$ via applying the bounding technique~\cite{S:77}. This, however, does not give poly$(n)$-time $(\epsilon,\delta)$-approximations of $P_f$ and $F_f$, for arbitrary $\epsilon>0$ and $0<\delta<1$. Using MCS, also, one may require exponentially many simulation runs to compute such approximations of $P_f$ and $F_f$ (see Section~\ref{subsec:TA2}). 

To tackle this problem, we use the ideas and algorithms from the previous section, with new bounds to achieve the desired approximation parameters. In particular, we propose a fast and accurate algorithm to approximate $F_f$, using the near-minimum cutsets of weight no greater than $\alpha\geq 1$ times the minimum cutset weight, for a proper choice of $\alpha$ upper bounded by $3+o(1)$ (i.e., using $N^{(\alpha)}\leq n^{2\alpha}=O(n^6)$ cutsets (by Lemma~\ref{lem:AlphaMinCutsNumber})). Note that for approximating $P_f$, the proper choice of $\alpha$, as was shown in~\cite{K:01}, is upper bounded by~$2$.  

\subsection{Proposed Algorithm}\label{subsec:PA2}
For an all-terminal reliability system, the inputs of the proposed algorithm are the failure rates $\{\lambda_i\}$, the repair rates $\{\mu_i\}$, and the approximation parameters $\epsilon>0$ and $0<\delta<1$. The algorithm proceeds in steps as follows:
\begin{itemize}
\item[0.] Initialization:
\begin{itemize}
\item[0.1] Find a min-cutset $\mathcal{C}^{*}$, and compute $w^{*}=w(\mathcal{C}^{*})$ and $p^{*}=p(\mathcal{C}^{*})$;
\item[0.2] Compute $\mu_{\text{min}}$, $\lambda_{\text{max}}$, $w_{\text{max}}$, $\mu$, $\mu^{*}=\sum_{i\in \mathcal{C}^{*}} \mu_i$, $s^{*} = \min\{\max\{w^{*}/w_{\text{max}},1\},m\}$, and $\rho=\mu_{\text{min}}s^{*}-\lambda_{\text{max}} (m-s^{*})$;
\end{itemize}
\item[1.] If $p^{*}>n^{-4}$: 
\begin{itemize}
\item[1.1] Compute $\tilde{F}^{\text{MC}}_f$ as the output of MCS for $S = {\lceil (\mu(2+\epsilon)\log 8)/ (p^{*}\rho\epsilon^2)\rceil}$ and $T = \lceil 12\log(1/\delta)\rceil$;
\item[1.2] Return $\tilde{F}^{\text{MC}}_f$.
\end{itemize} 
\item[2.] If $p^{*}\leq n^{-4}$:
\begin{itemize}
\item[2.1] Take $\xi = (\epsilon/2)(\rho/\mu)$
\item[2.2] Take $\gamma = (w^{*}/\log n) -2$, and $\alpha =1+(2/\gamma)+(\log((2(\gamma+2)(\mu-s^{*}\mu_{\text{min}}))/(\xi\gamma (\mu-\mu^{*}))))/(\gamma\log n)$
\item[2.3] Enumerate all $N^{(\alpha)}$ $\alpha$-min cutsets;
\item[2.4] Compute a $(\xi,\delta/2)$-approximation $\tilde{P}^{(\alpha)}$ of $P$ using the KLM estimator; 
\item[2.5] Compute a $(\xi,\delta/2)$-approximation $\tilde{P}^{(\alpha)}_f$ of $P_f$ using the KLM estimator; 
\item[2.6] Return $\tilde{F}_f = (\tilde{P}^{(\alpha)}_f-\tilde{P}^{(\alpha)})\mu$.\end{itemize} 
\end{itemize}

The running time of this algorithm is $O(n^{4}m(m+n)(1/\epsilon^2)\log(1/\delta))$ and $O(N^{(\alpha)} m^3(1/\epsilon^2)\log(1/\delta))$ for $p^{*}>n^{-4}$ and $p^{*}\leq n^{-4}$, respectively. (The running time of each step of the algorithm is given in Section~\ref{subsec:CC2}.)

Theoretically, $N^{(\alpha)} = O(n^{6})$ for our choice of $\alpha$ (see Lemma~\ref{lem:AlphaMinCutsNumber}). This result follows from a worst-case analysis. However, for many practical systems, e.g., Internet2 network (see Section~\ref{sec:SR}), $N^{(\alpha)}$ is much smaller, e.g., $O(n^2)$. The threshold $n^{-4}$ for $p^{*}$ in the algorithm is chosen to have a matching running time $O(n^8)$ in both cases of $p^{*}$ for systems with $N^{(\alpha)} = O(n^2)$. Nevertheless, for any choice of the threshold, the running time of the proposed algorithm is significantly less than that of MCS, for sufficiently small $p^{*}$ (depending on the threshold's choice). 

The details of the enumeration of $N^{(\alpha)}$ $\alpha$-min cutsets (Step 2.3) and the computations of $\tilde{P}^{(\alpha)}_f$ and $\tilde{P}^{(\alpha)}$ (Steps 2.4 and 2.5) are as follows. 

\subsubsection{Enumeration of $\alpha$-min Cutsets}\label{subsubsec:AlphaMinCutsets} We enumerate $\alpha$-min cutsets by using a randomized algorithm, referred to as the  \emph{recursive generalized contraction (RGC) algorithm}, due to Karger and Stein~\cite{KS:96}. (The non-recursive and recursive original contraction algorithms output min-cutsets~\cite{K:93}.) For simplicity, we explain the non-recursive version of this algorithm.

The (non-recursive) generalized contraction algorithm proceeds in rounds. In each round, one component, say $i$, is randomly chosen with probability of choosing component $i$ equal to $w_i/w$, and the two end-nodes of the component $i$ are merged, while maintaining all components from either of these nodes to other nodes. The algorithm continues this process until more than $\lceil 2\alpha\rceil$ nodes remain, and terminates otherwise. Once terminated, the algorithm selects a cutset in the resulting (multi-) system at random, and returns this cutset. 

Run the RGC algorithm for $n^{2\alpha}\log n^{2\alpha+c}$ times for an arbitrary $c>0$ (in $O(n^{2\alpha}\log^2 n)$ time), and denote by $C^{(\alpha)}$ the set of all $N^{(\alpha)}$ (distinct) output cutsets.

\subsubsection{Computation of $\tilde{P}^{(\alpha)}$}\label{subsubsec:TildeP2} Define $Z_i$ and $P_Z(i)$ as in Section~\ref{subsubsec:TildeP1}, except using only the $N^{(\alpha)}$ cutsets in $C^{(\alpha)}$, instead of all the $N$ cutsets in $C$; and construct a DNF formula \[\Phi_{N^{(\alpha)}}=Z_1\vee \dots \vee Z_{N^{(\alpha)}}.\] Run the KLM estimator with inputs $(\Phi_{N^{(\alpha)}},P_Z;\xi,\delta/2)$, and denote by $\tilde{P}^{(\alpha)}$ the output.

\subsubsection{Computation of $\tilde{P}_f^{(\alpha)}$}\label{subsubsec:TildePf2}

Similar to computing $\tilde{P}^{(\alpha)}$, in order to compute $\tilde{P}^{(\alpha)}_f$, construct a DNF formula \[\Phi_{N^{(\alpha)},f}=X_1\vee\dots\vee X_{N^{(\alpha)}},\] where $X_i$ is defined as in Section~\ref{subsubsec:TildePf1}, except only for $i\in [N^{(\alpha)}]$, and not all $i\in [N]$. Define $P_X(i)$ accordingly as before. Run the KLM estimator with inputs $(\Phi_{N^{(\alpha)},f},P_X;\xi,\delta/2)$, and denote by $\tilde{P}^{(\alpha)}_f$ the output.

\subsection{Theoretical Analysis}\label{subsec:TA2}

\begin{theorem}\label{thm:Main2}
The output of the algorithm in Section~\ref{subsec:PA2}, $\tilde{F}^{\text{MC}}_f$ or $\tilde{F}_f$, is an $(\epsilon,\delta)$-approximation of $F_f$. 
\end{theorem}


\begin{proof}
For the case of $p^{*}>n^{-4}$, the algorithm resorts to the MCS algorithm. We show that $\tilde{F}^{\text{MC}}_f$ is an $(\epsilon,\delta)$-approximation of $F_f$. To this end, we use a version of the Chernoff's bound as follows. (The proof of this result can be found in the appendix.)

\begin{lemma}\label{lem:CH}
Let $I_1,\dots,I_T$ be $T$ independent and identically distributed random variables such that $0\leq I_t\leq 1$ for all $t$. Let $I\triangleq \sum_{t=1}^{T} I_t$. Then, for any $\epsilon>0$, 
\begin{equation}
\Pr\{I\geq (1+\epsilon)\mathbb{E}(I)\}\leq \exp\left(-\left(\frac{\epsilon^2}{2+\epsilon}\right)\mathbb{E}(I)\right),  	
\end{equation} and 
\begin{equation}
\Pr\{I\leq (1-\epsilon)\mathbb{E}(I)\}\leq \exp\left(-\left(\frac{\epsilon^2}{2}\right)\mathbb{E}(I)\right).  	
\end{equation}
\end{lemma}

By the definition of $\varphi_{s,t}$ in Step~2 of MCS, it is easy to see that $0\leq \rho \leq \varphi_{s,t}\leq \mu$ for all $s,t$ (by~\eqref{eq:MuIs} and~\eqref{eq:LambdaIs}). Let $\hat{\varphi}_{s,t}\triangleq \varphi_{s,t}/\mu$. Note that $\hat{\varphi}_{s,t}$ are independent and identically distributed random variables such that $0\leq \hat{\varphi}_{s,t}\leq 1$ for all $s,t$. Recall that $\varphi_t = (\sum_{s=1}^{S} \varphi_{s,t})/S$ (Step~5 of MCS). Let $\hat{\varphi}_t \triangleq \varphi_t/\mu$. By the result of Lemma~\ref{lem:CH}, for any $\epsilon>0$, we get $\Pr\{|\hat{\varphi}_t-\mathbb{E}(\hat{\varphi}_t)|\geq \epsilon \mathbb{E}(\hat{\varphi}_t)\}\leq 2\exp(-(\epsilon^2/(2+\epsilon))S\mathbb{E}(\hat{\varphi}_t))$ for all $t$. Note that $\mathbb{E}(\varphi_t)=F_f$ (by~\eqref{eq:PFMC} and~\eqref{eq:FFMC}), and accordingly, $\mathbb{E}(\hat{\varphi}_t) = F_f/\mu$. Moreover, $F_f\geq P_f\rho$ (by~\eqref{eq:PFFFPF}) and $P_f\geq p^{*}$. (The probability that all components in a cutset are unavailable, $P_f$, is lower bounded by the probability that all components in a given min-cutset are unavailable, $p^{*}$, i.e., $P_f\geq p^{*}$.) Thus, for all $t$,
\begin{equation*}\label{eq:Prphi1}
\Pr\left\{|\varphi_t-F_f|\geq \epsilon F_f\right\}\leq 2\exp\left(-\left(\frac{\epsilon^2}{2+\epsilon}\right) \left(\frac{\rho}{\mu}\right) S p^{*}\right)	,
\end{equation*} or equivalently, 
\begin{equation*}
\Pr\left\{|\varphi_t-F_f|\geq \epsilon F_f\right\}\leq \frac{1}{4},	
\end{equation*} for the choice of $S$ in the algorithm (Step~1.1). Since $\tilde{F}^{\text{MC}}_f$, the output of MCS, is the median of $\{\varphi_t\}_{t=1}^{T}$, by applying the result of Lemma~\ref{lem:CH}, it can be shown that 
\begin{align}\label{eq:FfMCE}
\Pr\left\{|\tilde{F}_f^{\text{MC}}-F_f|\geq \epsilon F_f\right\}& \leq \delta	
\end{align} for the choice of $T$ in the algorithm (Step~1.1). (Similar technique is used in the appendix as part of the proof of Lemma~\ref{lem:KLME}.) Thus, $\tilde{F}^{\text{MC}}_f$ is an $(\epsilon,\delta)$-approximation of $F_f$ (by~\eqref{eq:FfMCE}). (Similarly, MCS can compute such an approximation of $P_f$ in poly$(n)$ time.)
 






Now, consider the case of $p^{*}\leq n^{-4}$. Note that, in this case, $P_f$ can be arbitrarily small (in $n$ and $m$), and $S$ can be arbitrarily large. Thus, MCS cannot compute an $(\epsilon,\delta)$-approximation of $F_f$ in poly$(n)$ time. (A similar negative result holds for computing such an approximation of $P_f$.)

Fix an arbitrary $\alpha\geq 1$. The following results hold for $\alpha$-min cutsets. (The proofs are given in the appendix.) 

\begin{lemma}\label{lem:AlphaMinCutsNumber}
The number of $\alpha$-min cutsets, $N^{(\alpha)}$, is bounded from above by $n^{2\alpha}$.
\end{lemma} 

\begin{lemma}\label{lem:AlphaMinCutsetsEnum}
Running the RGC algorithm $n^{2\alpha} \log n^{2\alpha+c}$ times, for any $c>0$, one can enumerate all $\alpha$-min cutsets in $O(n^{2\alpha}\log^2 n)$ time, with probability at least $1-n^{-c}$.	
\end{lemma}

By the choice of the algorithm (Step~2.2), 
\begin{equation}\label{eq:gamma}
\gamma=\frac{w^{*}}{\log n}-2,	
\end{equation} i.e., $p^{*}=n^{-2-\gamma}$. (Since $p^{*}\leq n^{-4}$, it holds that $\gamma\geq 2$.) The probability that all components in a given cutset are unavailable is upper bounded by $p^{*}$, and the probability that they are all unexposed is upper bounded by $1-s^{*}(\mu_{\text{min}}/\mu)$. (Note that $p(\mathcal{C})\leq p^{*}$ and $\left|\mathcal{C}\right|\geq s^{*}$ for all cutsets $\mathcal{C}$.) Thus, the probability that all components in a given $\alpha$-min cutset are unavailable and unexposed is upper bounded by $p^{*}(1-s^{*}(\mu_{\text{min}}/\mu))$, or equivalently, $n^{-2-\gamma}(1-s^{*}(\mu_{\text{min}}/\mu))$. 

By Lemma~\ref{lem:AlphaMinCutsNumber}, there are at most $n^{2\alpha}$ $\alpha$-min cutsets. By applying union bound, the probability that all components in some $\alpha$-min cutset are unavailable and unexposed is upper bounded by $n^{-2-\gamma+2\alpha}(1-s^{*}(\mu_{\text{min}}/\mu))$. This result is generalizable for all cutsets of weight greater than $\alpha w^{*}$ as follows. (A related, yet weaker, result was shown in~\cite[Theorem 2.9]{K:01}.) 


Let $\{\mathcal{C}_{1},\dots,\mathcal{C}_{N-N^{(\alpha)}}\}$ be the set of all cutsets of weight greater than $\alpha w^{*}$. Assume that $w(\mathcal{C}_1)\leq \dots \leq w(\mathcal{C}_{N-N^{(\alpha)}})$. Let $M \triangleq \min\{n^{2\alpha}, N-N^{(\alpha)}\}$. First, consider the cutsets $\mathcal{C}_1,\dots,\mathcal{C}_M$. For any $1\leq j\leq M$, the probability that all components in $\mathcal{C}_j$ are unavailable and unexposed, $\exp(-w(\mathcal{C}_j)) (1-\sum_{i\in \mathcal{C}_j}\mu_i/\mu)$, is upper bounded by $(p^{*})^{\alpha}(1-s^{*}(\mu_{\text{min}}/\mu))$. By applying union bound, the probability that all components in $\mathcal{C}_j$ for some $1\leq j\leq n^{2\alpha}$ are unavailable and unexposed is upper bounded by 
\begin{equation}\label{eq:Fact1}
n^{2\alpha} (p^{*})^{\alpha} \left(1- \frac{s^{*}\mu_{\text{min}}}{\mu}\right)=n^{-\alpha\gamma}\left(1- \frac{s^{*}\mu_{\text{min}}}{\mu}\right).
\end{equation}

Next, consider the remainder of the cutsets $\mathcal{C}_{M+1},\dots,\mathcal{C}_{N-N^{(\alpha)}}$ (if any). For any $\beta>0$, the number of cutsets of weight less than or equal to $\beta w^{*}$ is upper bounded by $n^{2\beta}$ (by Lemma~\ref{lem:AlphaMinCutsNumber}). Thus, $w(\mathcal{C}_{n^{2\beta}})\geq \beta w^{*}$. This gives $w(\mathcal{C}_j)\geq (w^{*} \log j)/(2\log n)$ for all $j$, and subsequently, the probability that all components in cutset $\mathcal{C}_j$ are unavailable and unexposed is upper bounded by $(p^{*})^{(\log j)/(2\log n)}(1-s^{*}(\mu_{\text{min}}/\mu))=j^{-1-\gamma/2}(1-s^{*}(\mu_{\text{min}}/\mu))$. Again by a union-bound analysis, the probability that all components in cutset $\mathcal{C}_j$ for some $j>n^{2\alpha}$ are unavailable and unexposed is upper bounded by 
\begin{dmath}\label{eq:Fact2}
\sum_{j>n^{2\alpha}} j^{-1-\gamma/2}\left(\hspace{-0.125em}1- \frac{s^{*}\mu_{\text{min}}}{\mu}\hspace{-0.125em}\right)\hspace{-0.125em}\leq \frac{2}{\gamma} n^{-\alpha\gamma}\hspace{-0.125em}\left(\hspace{-0.125em}1- \frac{s^{*}\mu_{\text{min}}}{\mu}\hspace{-0.125em}\right).	
\end{dmath}

Putting~\eqref{eq:Fact1} and~\eqref{eq:Fact2} together, the probability that all components in some cutset $\mathcal{C}_j$, $j\in [N-N^{(\alpha)}]$, are unavailable and unexposed is upper bounded by 
\begin{dmath*}
n^{-\alpha\gamma}\left(1-\frac{s^{*}\mu_{\text{min}}}{\mu}\right)+ \frac{2}{\gamma} n^{-\alpha\gamma}\left(1-\frac{s^{*}\mu_{\text{min}}}{\mu}\right) = n^{-\alpha\gamma}\left(1+\frac{2}{\gamma}\right)\left(1-\frac{s^{*}\mu_{\text{min}}}{\mu}\right).
\end{dmath*} This immediately yields the following result.  

\begin{lemma}\label{lem:StrongCutsets}
The probability that all components in some cutset of weight greater than $\alpha w^{*}$ are unavailable and unexposed is bounded from above by 
\begin{equation}\label{eq:weightgreaterprob}
n^{-\alpha\gamma}\left(1+\frac{2}{\gamma}\right)\left(1-\frac{s^{*}\mu_{\text{min}}}{\mu}\right).	
\end{equation}
\end{lemma}



The probability that all components in some cutset are unavailable and unexposed, $P$, is lower bounded by the probability that all components in the min-cutset $\mathcal{C}^{*}$, found in Step~0.1, are unavailable and unexposed, $p^{*}(1-\mu^{*}/\mu)$. Note that this argument holds for any cutset $\mathcal{C}$, but the lower bound needs to be replaced with $p(\mathcal{C})(1-\sum_{i\in \mathcal{C}}\mu_i/\mu)$. Thus, 
\begin{equation}\label{eq:Ppstarfactor}
P\geq p^{*}\left(1-\frac{\mu^{*}}{\mu}\right).	
\end{equation} (Note that the bound in~\eqref{eq:Ppstarfactor} can be improved as follows. Enumerate all $N^{(1)}\leq n^2$ min-cutsets, e.g., by using the RGC algorithm in $O(n^2\log^2 n)$ time, and select a min-cut $\mathcal{C}$ with minimum $\sum_{i\in \mathcal{C}} \mu_i$, e.g., via sorting in $O(n^2\log n)$ time.) 

By combining~\eqref{eq:weightgreaterprob} and~\eqref{eq:Ppstarfactor}, it is easy to see that 
\begin{equation}\label{eq:Pbigeq}
n^{-\alpha\gamma}\left(1+\frac{2}{\gamma}\right)\left(1-\frac{s^{*}\mu_{\text{min}}}{\mu}\right)\leq \frac{\xi}{2} P
\end{equation} for any $\xi>0$ so long as 
\begin{equation}\label{eq:AlphaBound}
\alpha\geq \frac{1}{\gamma}\left(\gamma+2+\frac{\log\left(\frac{2}{\xi}\left(\frac{\gamma+2}{\gamma}\right)\left(\frac{\mu-s^{*}\mu_{\text{min}}}{\mu - \mu^{*}}\right)\right)}{\log n}\right).
\end{equation} By the choice of $\alpha$ in the algorithm (Step~2.2), it follows that~\eqref{eq:AlphaBound}, and consequently,~\eqref{eq:Pbigeq} hold. 


Let $P^{(\alpha)}$ be the probability that all components in some $\alpha$-min cutset are unavailable and unexposed. Then, 
\begin{equation}\label{eq:PMBound}
\left(1-\frac{\xi}{2}\right)P \stackrel{(a)}{\leq} P^{(\alpha)}\stackrel{(b)}{\leq} P,
\end{equation} where $(a)$ follows from \eqref{eq:Pbigeq}, and $(b)$ follows from the definitions of $P^{(\alpha)}$ and $P$. Thus, $P^{(\alpha)}$ is a $(\xi/2,0)$-approximation of $P$. Note that $P^{(\alpha)}$ corresponds to the $\alpha$-min cutsets, and there are poly$(n)$ such cutsets (Lemma~\ref{lem:AlphaMinCutsNumber}), and they can be enumerated in poly$(n)$ time (Lemma~\ref{lem:AlphaMinCutsetsEnum}). However, one cannot compute $P^{(\alpha)}$ in poly$(n)$ time via the inclusion-exclusion formula due to the exponential number of terms. We, instead, approximate $P^{(\alpha)}$ in poly$(n)$ time using the KLM estimator. 

The formula $\Phi_{N^{(\alpha)}}$ is true so long as some clause $Z_i$ is true. Thus, the truth probability of $\Phi_{N^{(\alpha)}}$ is equal to $P^{(\alpha)}$. Since $\Phi_{N^{(\alpha)}}$ is a DNF formula, the output of the KLM estimator with inputs $(\Phi_{N^{(\alpha)}},P_Z;\xi/2,\delta/2)$, denoted by $\tilde{P}^{(\alpha)}$, is a $(\xi/2,\delta/2)$-approximation of $P^{(\alpha)}$, i.e., 
\begin{equation}\label{eq:TildePAlphaM1stBound}
\Pr\left\{|\tilde{P}^{(\alpha)}-P^{(\alpha)}|\geq \frac{\xi}{2} P^{(\alpha)}\right\}\leq \frac{\delta}{2}.
\end{equation} By using~\eqref{eq:PMBound} and~\eqref{eq:TildePAlphaM1stBound}, we get  
\begin{equation}\label{eq:TildePM2ndBound}
\Pr\left\{|\tilde{P}^{(\alpha)}-P|\geq \xi P \right\}\leq \frac{\delta}{2}.
\end{equation} This immediately yields the following result.

\begin{lemma}\label{lem:Palphatilde}
$\tilde{P}^{(\alpha)}$ is a $(\xi,{\delta}/{2})$-approximation of $P$.
\end{lemma}

Let $P_f^{(\alpha)}$ be the probability that all components in some $\alpha$-min cutset are unavailable. Similar to~\eqref{eq:PMBound}, as was shown in~\cite{K:01}, it follows that
\begin{equation}\label{eq:PfMBound}
\left(1-\frac{\xi}{2}\right)P_f \leq P_f^{(\alpha)}\leq P_f.
\end{equation} 

The probability that the formula $\Phi_{N^{(\alpha)},f}$ is true is equal to $P^{(\alpha)}_f$. Thus, the output of the KLM estimator with inputs $(\Phi_{N^{(\alpha)},f},P_X;\xi/2,\delta/2)$, denoted by $\tilde{P}_f^{(\alpha)}$, is a $({\xi}/{2},{\delta}/{2})$-approximation of $P^{(\alpha)}_f$. Thus, 
\begin{equation}\label{eq:TildePfM1stBound}
\Pr\left\{|\tilde{P}_f^{(\alpha)}-P_f^{(\alpha)}|\geq \frac{\xi}{2} P_f^{(\alpha)} \right\}\leq \frac{\delta}{2}.
\end{equation} By using~\eqref{eq:PfMBound} and~\eqref{eq:TildePfM1stBound}, we get
\begin{equation}\label{eq:TildePfM2ndBound}
\Pr\left\{|\tilde{P}_f^{(\alpha)}-P_f|\geq \xi P_f \right\}\leq \frac{\delta}{2}.
\end{equation} Then, the following result is immediate. 

\begin{lemma}\label{lem:Palphaftilde}
$\tilde{P}_f^{(\alpha)}$ is a $(\xi,{\delta}/{2})$-approximation of $P_f$.
\end{lemma}

Putting together~\eqref{eq:TildePM2ndBound} and~\eqref{eq:TildePfM2ndBound}, it follows that
\begin{equation*}\label{eq:PPfBound}
\Pr\left\{\left|(\tilde{P}_f^{(\alpha)}-\tilde{P}^{(\alpha)})-\left(P_f-P\right)\right|\geq \xi (P_f+P) \right\}\leq \delta.
\end{equation*} The rest of the proof is the same as that in the proof of Theorem~\ref{thm:Main1} (and hence omitted), except that $\tilde{P}_f$ and $\tilde{P}$ are replaced with $\tilde{P}^{(\alpha)}_f$ and $\tilde{P}^{(\alpha)}$.\qed
\end{proof}

Similarly as in~\eqref{eq:FfMCE}, it can be shown that
\begin{align}\label{eq:AE}
\Pr\left\{|\tilde{F}_f^{\text{MC}}-F_f|\geq \epsilon\right\}& \leq \delta	
\end{align} so long as 
\begin{equation}\label{eq:TAE}
S\geq \mu\left(\frac{2F_f+\epsilon}{\epsilon^2}\right)\log 8
\end{equation} and 
\begin{equation}\label{eq:TAE2}
T\geq 12\log(1/\delta). 	
\end{equation}
Taking $S = \lceil (\mu(2\mu+\epsilon)\log 8)/\epsilon^2\rceil$ (noting $F_f\leq P_f\mu\leq \mu$ (by~\eqref{eq:PFFFPF})) and $T = \lceil12\log(1/\delta)\rceil$, MCS can approximate $F_f$, with an additive error of at most $\epsilon$ and an error probability of at most $\delta$, in poly$(n)$ time (by~\eqref{eq:AE}-\eqref{eq:TAE2}). (Since $\mu=O(m)=O(n^2)$, then $ST=O((n^4/\epsilon^2+n^2/\epsilon)\log(1/\delta))$.) Similar result holds for approximating $P_f$ using MCS. Note, however, that MCS cannot compute approximations of $P_f$ and $F_f$ within a multiplicative error in poly$(n)$ time. This suggests that approximating $P_f$ and $F_f$ within a multiplicative error is more challenging than that within an additive error. 

\begin{figure}
\centering
\includegraphics[width=0.49\textwidth]{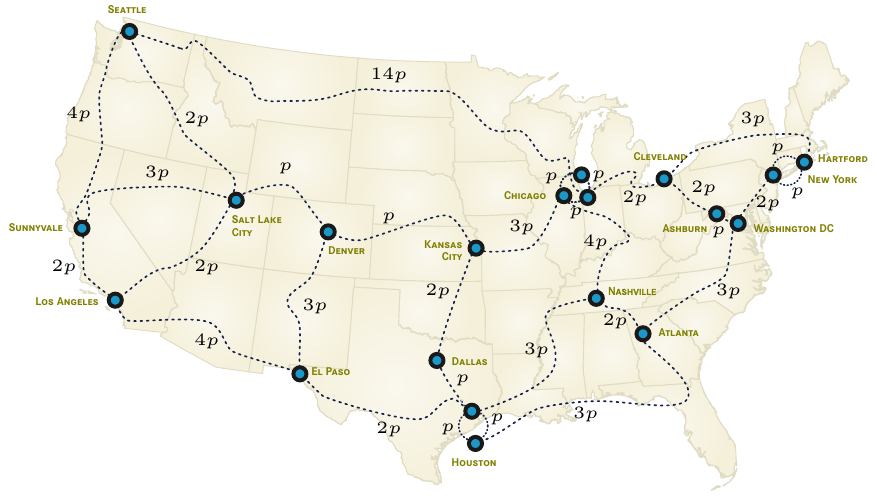}\vspace{0.25cm}
\caption{Layer 3 of Internet2 network (research and education IP and peering IP) with $20$ nodes/terminals and $31$ links.}\label{fig:System2}
\end{figure}

\subsection{Computational Complexity}\label{subsec:CC2}
The initialization (Step~0) can be run in $O(n^3)$ time. Step~0.1 can be run in $O(n^3)$ time, e.g., by finding a min-cutset in $O(nm+n^2\log n)$ time by the Stoer-Wagner algorithm~\cite{SW:94}, and computing $p^{*}$ and $w^{*}$ in $O(m)$ time; In Step~0.2, $\mu_{\text{min}}$, $\lambda_{\text{max}}$, and $w_{\text{max}}$ can be computed in $O(m\log m)$ time (via sorting), $\mu$ and $\mu^{*}$ in $O(m)$ time, and $s^{*}$ and $\rho$ in $O(1)$ time. 


For simplicity, we ignore $\epsilon$ and $\delta$ in what follows, since they do not depend on $n$ and $m$. In Step~1, $\tilde{F}_f^{\text{MC}}$ can be computed in $O(n^{4}m(m+n))=O(n^8)$ time by using MCS. Step~2 can also be run in $O(n^{12})$ time as follows. (These running times follow from a worst-case analysis, and some of the bounds are far from tight.) Since $\gamma\geq 2$ and $\xi=\Omega(n^{-2})$ for our choice of $\xi$ (Step~2.1), it is easy to see that our choice of $\alpha$ (Step~2.2) is upper bounded by $3+O(1/\log n)$, or equivalently, $3+o(1)$. Moreover, $N^{(\alpha)}\leq n^{2\alpha}$ (by Lemma~\ref{lem:AlphaMinCutsNumber}). Thus, in Step~2.3, it suffices to enumerate $N^{(\alpha)}=O(n^6)$ $\alpha$-min cutsets for the choice of $\alpha$ in the algorithm. (For computing $\tilde{P}_f^{(\alpha)}$, as was shown in~\cite{K:01}, it suffices to enumerate $\alpha$-min cutsets for some $\alpha\leq 2$, and there are $O(n^4)$ such cutsets.) Step~2.3 can be run, e.g., using the RGC algorithm~\cite{KS:96}, in $O(n^{2\alpha}\log^2 n)=O(n^6\log^2 n)$ time (by Lemma~\ref{lem:AlphaMinCutsetsEnum}). In Steps~2.4 and~2.5, the estimates $\tilde{P}^{(\alpha)}$ and $\tilde{P}^{(\alpha)}_f$ can be computed by running the (self-adjusting) KLM estimator in $O(N^{(\alpha)} m^3) = O(n^{12})$ time. From these arguments, we conclude that the proposed algorithm runs in poly$(n)$ time.

\section{Simulation Results}\label{sec:SR}
In this section, we compare the accuracy and the running time of the Monte Carlo simulation and the proposed algorithm for approximating $F_f$ for a $3\times 3$ grid network (Fig.~\ref{fig:System}) and for the Internet2 network~\cite{I2net}, particularly, the layer 3 (Fig.~\ref{fig:System2}) and the layer 2 (Fig.~\ref{fig:System3}). We refer to these systems as \emph{S1}, \emph{S2}, and \emph{S3}, respectively. S1, S2, and S3 contain $n=9$, $20$, and $35$ nodes and $m=12$, $31$, and $46$ components, respectively. 

\begin{figure}
\centering
\includegraphics[width=0.49\textwidth]{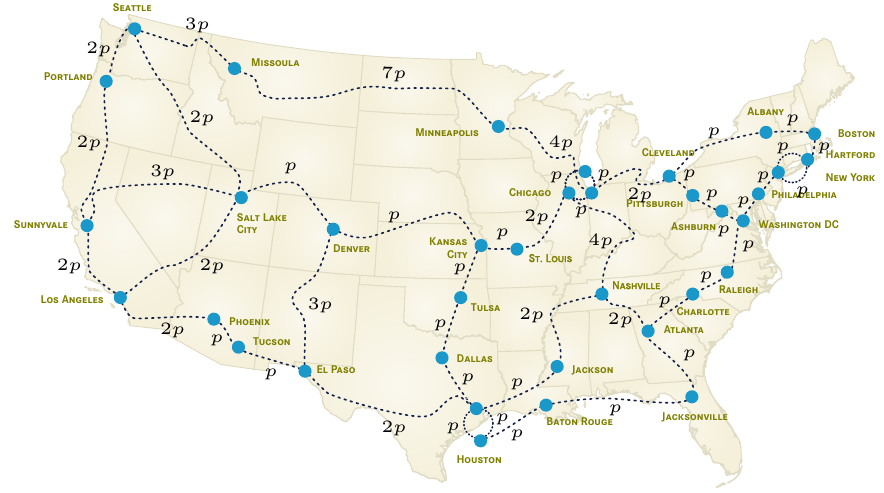}\vspace{0.25cm}
\caption{Layer 2 of Internet2 network (SDN Ethernet) with $35$ nodes/terminals and $46$ links.}\label{fig:System3}
\end{figure}

For simulations, we have used the Matlab Parallel Computing toolbox for 32 workers on CentOS linux machines with 2X - Intel Xeon E5-2697A V4 @ 2.6GHz processors and 512GB memory. 

\begin{table*}[t]
\centering
\caption{Parameters of the Systems S1, S2, and S3}
\label{tab:star}
\begin{tabular}{|c|c|c|c!{\vrule width 1pt}c|c|c|c!{\vrule width 1pt}c|c|c|c|}
\hline
\multicolumn{4}{|c!{\vrule width 1pt}}{S1} & \multicolumn{4}{c!{\vrule width 1pt}}{S2} & \multicolumn{4}{c|}{S3}\\ \hline
$p$ & $p^{*}$ & $\alpha$ & $N^{(\alpha)}$ & $p$ & $p^{*}$ & $\alpha$ & $N^{(\alpha)}$ & $p$ & $p^{*}$ & $\alpha$ & $N^{(\alpha)}$ \\ \noalign{\hrule height 1pt}
\num{1e-2} & \num{1.00e-4} & \num{2.93} & \num{53} & \num{1e-3} & \num{6.00e-6} & \num{2.71} & \num{167} & \num{1e-4} & \num{2.80e-7} & \num{2.39} & \num{329} \\ \hline
\num{1e-2.2} & \num{3.98e-5} & \num{2.60} & \num{53} & \num{1e-3.2} & \num{2.39e-6} & \num{2.50} & \num{116} & \num{1e-4.2} & \num{1.11e-7} & \num{2.26} & \num{204} \\ \hline
\num{1e-2.4} & \num{1.58e-5} & \num{2.40} & \num{37} & \num{1e-3.4} & \num{9.51e-7} & \num{2.35} & \num{79} & \num{1e-4.4} & \num{4.44e-8} & \num{2.16} & \num{146} \\ \hline
\num{1e-2.6} & \num{6.31e-6} & \num{2.22} & \num{37} & \num{1e-3.6} & \num{3.79e-7} & \num{2.20} & \num{76} & \num{1e-4.6} & \num{1.77e-8} & \num{2.06} & \num{139} \\ \hline
\num{1e-2.8} & \num{2.51e-6} & \num{2.09} & \num{37} & \num{1e-3.8} & \num{1.51e-7} & \num{2.09} & \num{61} & \num{1e-4.8} & \num{7.03e-9} & \num{1.97} & \num{139} \\ \hline
\num{1e-3} & \num{1.00e-6} & \num{2.00} & \num{20} & \num{1e-4} & \num{6.00e-8} & \num{2.01} & \num{48} & \num{1e-5} & \num{2.80e-9} & \num{1.90} & \num{139}\\ \hline
\num{1e-3.2} & \num{3.98e-7} & \num{1.92} & \num{20} & \num{1e-4.2} & \num{2.39e-8} & \num{1.94} & \num{36} & \num{1e-5.2} & \num{1.11e-9} & \num{1.84} & \num{139} \\ \hline
\num{1e-3.4} & \num{1.58e-7} & \num{1.84} & \num{20} & \num{1e-4.4} & \num{9.51e-9} & \num{1.87} & \num{31} & \num{1e-5.4} & \num{4.44e-10} & \num{1.78} & \num{139} \\ \hline
\num{1e-3.6} & \num{6.31e-8} & \num{1.78} & \num{20} & \num{1e-4.6} & \num{3.79e-9} & \num{1.81} & \num{30} & \num{1e-5.6} & \num{1.77e-10} & \num{1.73} & \num{139}\\ \hline
\num{1e-3.8} & \num{2.51e-8} & \num{1.72} & \num{20} & \num{1e-4.8} & \num{1.51e-9} & \num{1.76} & \num{30} & \num{1e-5.8} & \num{7.03e-11} & \num{1.70} & \num{100} \\ \hline 
\end{tabular}
\end{table*}

\begin{table*}[t]
\centering
\caption{Approximations of Failure Frequency for the System S1 Using the Proposed Algorithm and the MCS Algorithm}
\label{tab:results1}
\begin{tabular}{|c!{\vrule width 1pt}c|c!{\vrule width 1pt}c|c!{\vrule width 1pt}c|c!{\vrule width 1pt}c|c!{\vrule width 1pt}c|c|}
\hline
& \multicolumn{2}{c!{\vrule width 1pt}}{Bounds} & \multicolumn{2}{c!{\vrule width 1pt}}{Approximation} &  \multicolumn{2}{c!{\vrule width 1pt}}{Running Time (sec)} & \multicolumn{2}{c!{\vrule width 1pt}}{Theoretical Error} & \multicolumn{2}{c|}{Actual Error} \\ \hline
$p$ & $F^{-}_f$ & $F^{+}_f$ & Proposed & MCS & Proposed & MCS & Proposed & MCS & Proposed & MCS \\ \noalign{\hrule height 1pt}
\num{1e-2}  & \num{8.46433e-4} & \num{8.48688e-4} & \num{8.47117e-4} & \num{8.42029e-4} & \num{220} & \num{369} & \num{0.36} & \num{0.30} & \num{1.86e-3} & \num{7.87e-3} \\ \hline
\num{1e-2.2}  & \num{3.30300e-4} & \num{3.30651e-4} & \num{3.30411e-4} & \num{3.28466e-4} & \num{220} & \num{370} & \num{0.36} & \num{0.48} & \num{7.27e-4} & \num{6.62e-3} \\ \hline
\num{1e-2.4}  & \num{1.29782e-4} & \num{1.29837e-4} & \num{1.29800e-4} & \num{1.30892e-4} & \num{219} & \num{366} & \num{0.29} & \num{0.81} & \num{2.85e-4} & \num{8.55e-3} \\ \hline
\num{1e-2.6}  & \num{5.12314e-5} & \num{5.12401e-5} & \num{5.12338e-5} & \num{5.04104e-5} & \num{217} & \num{354} & \num{0.29} & \num{1.41} & \num{1.23e-4} & \num{1.62e-2} \\ \hline
\num{1e-2.8}  & \num{2.02852e-5} & \num{2.02866e-5} & \num{2.02855e-5} & \num{2.08369e-5} & \num{220} & \num{349} & \num{0.29} & \num{2.59} & \num{5.42e-5} & \num{2.72e-2} \\ \hline
\num{1e-3}  & \num{8.04785e-6} & \num{8.04807e-6} & \num{8.04784e-6} & \num{8.13261e-6} & \num{181} & \num{289} & \num{0.23} & \num{5.95} & \num{2.81e-5} & \num{1.05e-2} \\ \hline
\num{1e-3.2}  & \num{3.19689e-6} & \num{3.19693e-6} & \num{3.19690e-6} & \num{3.19932e-6} & \num{216} & \num{340} & \num{0.21} & \num{10.9} & \num{8.04e-6} & \num{7.59e-4} \\ \hline
\num{1e-3.4}  & \num{1.27094e-6} & \num{1.27094e-6} & \num{1.27094e-6} & \num{1.16440e-6} & \num{213} & \num{334} & \num{0.21} & \num{24.9} & \num{3.51e-6} & \num{8.38e-2} \\ \hline
\num{1e-3.6}  & \num{5.05526e-7} & \num{5.05527e-7} & \num{5.05527e-7} & \num{5.82654e-7} & \num{213} & \num{334} & \num{0.21} & \num{59.7} & \num{1.50e-6} & \num{1.53e-1} \\ \hline
\num{1e-3.8}  & \num{2.01142e-7} & \num{2.01142e-7} & \num{2.01142e-7} & \num{0} & \num{214} & \num{333} & \num{0.21} & $>100$ & \num{4.65e-7} & -- \\ \hline
\end{tabular}
\end{table*}

The unavailability probability of all components in the system S1 is assumed to be $p$ for different values of $p\in \{10^{-2},10^{-2.2},\dots,10^{-3.8}\}$. The unavailability probability of each link (component) in the systems S2 and S3 is assumed to be proportional to the number of hops between the two end-points of the link in the layer 1 (see, for details,~\cite{I2net}). For example, for a given $p$, the unavailability probability of the link between Seattle and Salt Lake City in the system S2 is assumed to be $2p$ since they are two hops away in the layer 1, and the unavailability probability of the link between Denver and El Paso in the system S3 is assumed to be $3p$ since they are three hops away in the layer 1. For the systems S2 and S3, we consider different values of $p\in \{10^{-3},10^{-3.2},\dots,10^{-4.8}\}$ and $p\in\{10^{-4},10^{-4.2},\dots,10^{-5.8}\}$, respectively. Note that we consider cases with very small values of $p\ll 1$ since the failure of the system is a rare event in such cases. 

For each component with unavailability probability $\theta p$ (for some integer $\theta$), the repair rate and the failure rate of the component are assumed to be $1$ and $\theta p/(1-\theta p)$, respectively. 



Table~\ref{tab:star} lists the maximum failure probability of a cutset ($p^{*}$), the parameter $\alpha$, and the number of $\alpha$-min cutsets being enumerated ($N^{(\alpha)}$), for each of the systems S1, S2, and S3 and each $p$. To enumerate the $\alpha$-min cutsets, for each system and each $p$, we ran the RGC algorithm $2\times 10^4$ times, and the total running time was about $30$ seconds for each case.

\begin{table*}[t]
\centering
\caption{Approximations of Failure Frequency for the System S2 Using the Proposed Algorithm and the MCS Algorithm}
\label{tab:results2}
\begin{tabular}{|c!{\vrule width 1pt}c|c!{\vrule width 1pt}c|c!{\vrule width 1pt}c|c!{\vrule width 1pt}c|c!{\vrule width 1pt}c|c|}
\hline
& \multicolumn{2}{c!{\vrule width 1pt}}{Bounds} & \multicolumn{2}{c!{\vrule width 1pt}}{Approximation} &  \multicolumn{2}{c!{\vrule width 1pt}}{Running Time (sec)} & \multicolumn{2}{c!{\vrule width 1pt}}{Theoretical Error} & \multicolumn{2}{c|}{Actual Error} \\ \hline
$p$ & $F^{-}_f$ & $F^{+}_f$ & Proposed & MCS & Proposed & MCS & Proposed & MCS & Proposed & MCS \\ \noalign{\hrule height 1pt}
\hspace{-0.25em}\num{1e-3}\hspace{-0.25em} & \num{3.32832e-5} & \num{3.32870e-5} & \num{3.32832e-05} & \num{3.33671e-05} & \num{643} & \num{652} & \num{1.85} & \num{2.12} & \num{1.15e-4} & \num{2.52e-3} \\ \hline
\hspace{-0.25em}\num{1e-3.2}\hspace{-0.25em} & \num{1.30611e-5} & \num{1.30617e-5} & \num{1.30612e-5} & \num{1.30204e-5} & \num{599} & \num{628} & \num{1.41} & \num{3.82} & \num{3.61e-5} & \num{3.16e-3} \\ \hline
\hspace{-0.25em}\num{1e-3.4}\hspace{-0.25em} & \num{5.15235e-6} & \num{5.15244e-6} & \num{5.15244e-6} & \num{4.94345e-6} & \num{520} & \num{631} & \num{1.11} & \num{7.61} & \num{1.80e-5} & \num{4.06e-3} \\ \hline
\hspace{-0.25em}\num{1e-3.6}\hspace{-0.25em} & \num{2.03932e-6} & \num{2.03933e-6} & \num{2.03933e-6} & \num{2.08211e-6} & \num{513} & \num{685} & \num{1.05} & \num{16.4} & \num{6.20e-6} & \num{2.10e-2} \\ \hline
\hspace{-0.25em}\num{1e-3.8}\hspace{-0.25em} & \num{8.08888e-7}  & \num{8.08890e-7} & \num{8.08890e-7} & \num{6.96292e-7} & \num{488} & \num{684} & \num{0.93} & \num{38.0} & \num{2.07e-6} & \num{1.39e-1} \\ \hline
\hspace{-0.25em}\num{1e-4}\hspace{-0.25em} & \num{3.21277e-7} & \num{3.21277e-7} & \num{3.21279e-07} & \num{3.48829e-7} & \num{482} & \num{655} & \num{0.81} & \num{91.6} & \num{7.54e-6} & \num{8.58e-2} \\ \hline
\hspace{-0.25em}\num{1e-4.2}\hspace{-0.25em} & \num{1.27715e-7} & \num{1.27715e-7} & \num{1.27715e-7} & \num{0} & \num{459} & \num{656} & \num{0.70} & $>10^2$ & \num{5.06e-7} & -- \\ \hline
\hspace{-0.25em}\num{1e-4.4}\hspace{-0.25em} & \num{5.07971e-8} & \num{5.07971e-8} & \num{5.07971e-8} & \num{0} & \num{461} & \num{648} & \num{0.64} & $>10^2$ & \num{4.43e-7} & -- \\ \hline
\hspace{-0.25em}\num{1e-4.6}\hspace{-0.25em} & \num{2.02109e-8} & \num{2.02109e-8} & \num{2.02109e-8} & \num{0} & \num{464} & \num{668} & \num{0.63} & $>10^3$ & \num{2.16e-7} & -- \\ \hline
\hspace{-0.25em}\num{1e-4.8}\hspace{-0.25em} & \num{8.04312e-9} & \num{8.04312e-9} & \num{8.04312e-9} & \num{0} & \num{469} & \num{656} & \num{0.63} & $>10^3$ & \num{8.60e-8} & -- \\ \hline
\end{tabular}
\end{table*}

\begin{table*}[t]
\centering
\caption{Approximations of Failure Frequency for the System S3 Using the Proposed Algorithm and the MCS Algorithm}
\label{tab:results3}
\begin{tabular}{|c!{\vrule width 1pt}c|c!{\vrule width 1pt}c|c!{\vrule width 1pt}c|c!{\vrule width 1pt}c|c!{\vrule width 1pt}c|c|}
\hline
& \multicolumn{2}{c!{\vrule width 1pt}}{Bounds} & \multicolumn{2}{c!{\vrule width 1pt}}{Approximation} &  \multicolumn{2}{c!{\vrule width 1pt}}{Running Time (sec)} & \multicolumn{2}{c!{\vrule width 1pt}}{Theoretical Error} & \multicolumn{2}{c|}{Actual Error} \\ \hline
$p$ & $F^{-}_f$ & $F^{+}_f$ & Proposed & MCS & Proposed & MCS & Proposed & MCS & Proposed & MCS \\ \noalign{\hrule height 1pt}
\hspace{-0.25em} \num{1e-4}\hspace{-0.25em} & \hspace{-0.25em}\num{2.02035e-6}\hspace{-0.25em} & \hspace{-0.25em}\num{2.02128e-6}\hspace{-0.25em} & \num{2.02061e-6} & \hspace{-0.25em}\num{1.96230e-6}\hspace{-0.25em} & \hspace{-0.25em}\num{728}\hspace{-0.25em} & \hspace{-0.25em}\num{1002}\hspace{-0.25em} & \hspace{-0.25em}\num{3.24}\hspace{-0.25em} & \hspace{-0.25em}\num{20.0}\hspace{-0.25em} & \hspace{-0.25em}\num{3.30e-4}\hspace{-0.25em} & \hspace{-0.25em}\num{2.92e-2}\hspace{-0.25em} \\ \hline
\hspace{-0.25em}\num{1e-4.2}\hspace{-0.25em} & \hspace{-0.25em}\num{8.04264e-7}\hspace{-0.25em} & \hspace{-0.25em}\num{8.04497e-7}\hspace{-0.25em} & \hspace{-0.25em}\num{8.04321e-7}\hspace{-0.25em} & \hspace{-0.25em}\num{7.01045e-7}\hspace{-0.25em} & \hspace{-0.25em}\num{589}\hspace{-0.25em} & \hspace{-0.25em}\num{935}\hspace{-0.25em} & \hspace{-0.25em}\num{2.62}\hspace{-0.25em} & \hspace{-0.25em}\num{50.4}\hspace{-0.25em} & \hspace{-0.25em}\num{2.19e-4}\hspace{-0.25em} & \hspace{-0.25em}\num{1.29e-1}\hspace{-0.25em} \\ \hline
\hspace{-0.25em}\num{1e-4.4}\hspace{-0.25em} & \hspace{-0.25em}\num{3.20170e-7}\hspace{-0.25em} & \hspace{-0.25em}\num{3.20229e-7}\hspace{-0.25em} & \hspace{-0.25em}\num{3.20203e-7}\hspace{-0.25em} & \num{4.36921e-7} & \hspace{-0.25em}\num{569}\hspace{-0.25em} & \hspace{-0.25em}\num{999}\hspace{-0.25em} & \hspace{-0.25em}\num{2.13}\hspace{-0.25em} & \hspace{-0.25em}$>10^2$\hspace{-0.25em} & \hspace{-0.25em}\num{1.02e-4}\hspace{-0.25em} & \hspace{-0.25em}\num{3.65e-1}\hspace{-0.25em} \\ \hline
\hspace{-0.25em}\num{1e-4.6}\hspace{-0.25em} & \hspace{-0.25em}\num{1.27459e-7}\hspace{-0.25em} & \hspace{-0.25em}\num{1.27474e-7}\hspace{-0.25em} & \hspace{-0.25em}\num{1.27468e-7}\hspace{-0.25em} & \hspace{-0.25em}\num{0}\hspace{-0.25em} & \hspace{-0.25em}\num{560}\hspace{-0.25em} & \hspace{-0.25em}\num{1024}\hspace{-0.25em} & \hspace{-0.25em}\num{2.08}\hspace{-0.25em} & \hspace{-0.25em}$>10^2$\hspace{-0.25em} & \hspace{-0.25em}\num{7.17e-5}\hspace{-0.25em} & -- \\ \hline
\hspace{-0.25em}\num{1e-4.8}\hspace{-0.25em} & \hspace{-0.25em}\num{5.07415e-8}\hspace{-0.25em} & \hspace{-0.25em}\num{5.07452e-8}\hspace{-0.25em} & \hspace{-0.25em}\num{5.07436e-8}\hspace{-0.25em} & \hspace{-0.25em}\num{0}\hspace{-0.25em} & \hspace{-0.25em}\num{554}\hspace{-0.25em} & \hspace{-0.25em}\num{1009}\hspace{-0.25em} & \hspace{-0.25em}\num{2.07}\hspace{-0.25em} & \hspace{-0.25em}$>10^2$\hspace{-0.25em} & \hspace{-0.25em}\num{4.14e-5}\hspace{-0.25em} & -- \\ \hline
\hspace{-0.25em}\num{1e-5}\hspace{-0.25em} & \hspace{-0.25em}\num{2.02003e-8}\hspace{-0.25em} & \hspace{-0.25em}\num{2.02013e-8}\hspace{-0.25em} & \hspace{-0.25em}\num{2.02007e-8}\hspace{-0.25em} & \hspace{-0.25em}\num{0}\hspace{-0.25em} & \hspace{-0.25em}\num{555}\hspace{-0.25em} & \hspace{-0.25em}\num{976}\hspace{-0.25em} & \hspace{-0.25em}\num{2.07}\hspace{-0.25em} & \hspace{-0.25em}$>10^3$\hspace{-0.25em} & \hspace{-0.25em}\num{2.77e-5}\hspace{-0.25em} & -- \\ \hline
\hspace{-0.25em}\num{1e-5.2}\hspace{-0.25em} & \hspace{-0.25em}\num{8.04185e-9}\hspace{-0.25em} & \hspace{-0.25em}\num{8.04209e-9}\hspace{-0.25em} & \hspace{-0.25em}\num{8.04175e-9}\hspace{-0.25em} & \hspace{-0.25em}\num{0}\hspace{-0.25em} & \hspace{-0.25em}\num{549}\hspace{-0.25em} & \hspace{-0.25em}\num{1007}\hspace{-0.25em} & \hspace{-0.25em}\num{2.07}\hspace{-0.25em} & \hspace{-0.25em}$>10^3$\hspace{-0.25em} & \hspace{-0.25em}\num{4.13e-5}\hspace{-0.25em} & -- \\ \hline
\hspace{-0.25em}\num{1e-5.4}\hspace{-0.25em} & \hspace{-0.25em}\num{3.20151e-9}\hspace{-0.25em} & \hspace{-0.25em}\num{3.20156e-9}\hspace{-0.25em} & \hspace{-0.25em}\num{3.20149e-9}\hspace{-0.25em} & \hspace{-0.25em}\num{0}\hspace{-0.25em} & \hspace{-0.25em}\num{551}\hspace{-0.25em} & \hspace{-0.25em}\num{1022}\hspace{-0.25em} & \hspace{-0.25em}\num{2.07}\hspace{-0.25em} &\hspace{-0.25em}$>10^4$\hspace{-0.25em} & \hspace{-0.25em}\num{2.20e-5}\hspace{-0.25em} & -- \\ \hline
\hspace{-0.25em}\num{1e-5.6}\hspace{-0.25em} & \hspace{-0.25em}\num{1.27454e-9}\hspace{-0.25em} & \hspace{-0.25em}\num{1.27455e-9}\hspace{-0.25em} & \hspace{-0.25em}\num{1.27453e-9}\hspace{-0.25em} & \hspace{-0.25em}\num{0}\hspace{-0.25em} & \hspace{-0.25em}\num{549}\hspace{-0.25em} & \hspace{-0.25em}\num{973}\hspace{-0.25em} & \hspace{-0.25em}\num{2.07}\hspace{-0.25em} & \hspace{-0.25em}$>10^4$\hspace{-0.25em} & \hspace{-0.25em}\num{1.93e-5}\hspace{-0.25em} & -- \\ \hline
\hspace{-0.25em}\num{1e-5.8}\hspace{-0.25em} & \hspace{-0.25em}\num{5.07402e-10}\hspace{-0.25em} & \hspace{-0.25em}\num{5.07406e-10}\hspace{-0.25em} & \hspace{-0.25em}\num{5.07404e-10}\hspace{-0.25em} & \hspace{-0.25em}\num{0}\hspace{-0.25em} & \hspace{-0.25em}\num{506}\hspace{-0.25em} & \hspace{-0.25em}\num{988}\hspace{-0.25em} & \hspace{-0.25em}\num{1.75}\hspace{-0.25em} & \hspace{-0.25em}$>10^4$\hspace{-0.25em} & \hspace{-0.25em}\num{4.20e-6}\hspace{-0.25em} & -- \\ \hline
\end{tabular}
\end{table*}

For fixed $\delta = 10^{-2}$ (i.e., an error probability at most $1\%$), Table~\ref{tab:results1} represents the approximations $\tilde{F}_f$ and $\tilde{F}^{\text{MC}}_f$ of $F_f$ for the system S1 using the proposed algorithm and the MCS algorithm, the total running time for each algorithm, and the theoretical error factor and the actual (observed) error factor of each algorithm. In particular, the theoretical error factor ($\epsilon$) is computed based on the analysis in Section~\ref{subsec:TA2}, and the actual error factor, i.e., $\max\{|\varphi-F^{-}_f|/F^{-}_f,|\varphi-F^{+}_f|/F^{-}_f\}$, where $\varphi$ is $\tilde{F}_f$ or $\tilde{F}^{\text{MC}}_f$ for the proposed or the MCS algorithm, respectively, is computed based on the first-order upper- and lower-bounds $F_f^{+}$ and $F_f^{-}$ on $F_f$ resulting from the bounding technique. The values of $F_f^{+}$ and $F_f^{-}$ are also given for reference in Table~\ref{tab:results1}. Similarly, Tables~\ref{tab:results2} and~\ref{tab:results3} correspond to the results for the systems S2 and S3, respectively.

To compute $F_f^{+}$ and $F_f^{-}$, we enumerated all cutsets in each system. The accuracy of this technique, however, cannot be fairly compared with that of the other two algorithms since the running time for enumerating all cutsets in each system (about 1 week for the system S2 and 2 weeks for the system S3) was much larger than that of the other two algorithms. 

A simple comparison of the MCS algorithm and the proposed algorithm shows that not only is the latter technique faster than the former, but also it provides more accuracy, for each system and for any given $p$. For example, for the system S1 and $p=10^{-3}$, running the proposed algorithm for $181$ seconds one can approximate $F_f$ with an error factor of $2.81\times 10^{-5}$; whereas running the MCS algorithm for $289$ seconds, one can only approximate $F_f$ with an error factor of $1.05\times 10^{-2}$. As an another example, for the system S1 and $p=10^{-3.8}$, the proposed algorithm provides an approximation with an error factor of $4.65\times 10^{-7}$ in $214$ seconds; whereas running the MCS algorithm for $333$ seconds, no failure state is detected (and hence the output is ``zero,'' and the actual error factor is meaningless). Similar comparison results hold for larger systems S2 and S3 (see Tables~\ref{tab:results2} and~\ref{tab:results3}). Note that the advantages of the proposed algorithm over the MCS algorithm are even more evident for larger systems.

These results confirm that the proposed algorithm offers a significantly better tradeoff between running time and approximation accuracy. Note also that the improvements become more profound for smaller values of $p$. This is evident from the two examples above. Furthermore, for smaller values of $p$ the running time of the proposed algorithm becomes smaller, whereas the running time of the MCS algorithm remains almost the same. This comes from the fact that as $p$ decreases the parameter $\alpha$ decreases (even for smaller error factor), and the number of $\alpha$-min cutsets decreases. 

Furthermore, the theoretical error factors for both algorithms, obtained from the worst-case analysis in Section~\ref{subsec:TA2}, are much larger than the actual error factors, computed based on the simulation results. This suggests that for many practical systems each of these algorithms may provide a better complexity-accuracy tradeoff than what was shown in the analysis. However, the ratio of the actual error factor to the theoretical error factor for the proposed algorithm, when compared to that for the MCS algorithm, is much smaller. Thus, one can expect that the proposed algorithm outperforms the MCS algorithm in practice even more than that in theory.

\section{Conclusion and Open Problems}\label{sec:COP}
In this work, we considered the problem of estimating the failure frequency of large-scale composite $k$-terminal reliability systems. It was previously shown that the failure probability can be efficiently approximated with provable guarantees. However, no such result was previously known for approximating the failure frequency. 

We proposed the first polynomial-time (in the number of cutsets in the system) algorithm for approximating the failure frequency within an arbitrary multiplicative error and with an arbitrary error probability. The main ideas of the proposed algorithm are in summary as follows: (i)~the failure frequency can be linked to a linear combination of the probabilities of two sets of events, and (ii)~each of these probabilities can be written as a Boolean formula of disjunctive normal form, and the truth probability of such a formula can be estimated using an unbiased estimator within an arbitrary error factor and with an arbitrary error probability. 

The number of cutsets of a system can generally grow exponentially with the number of nodes in the system. Motivated by this, for the special case of all-terminal reliability systems in which all nodes are terminals, we proposed the first polynomial-time (in the number of nodes in the system) algorithm using only the near-minimum cutsets, instead of all cutsets, for approximating the failure frequency with an arbitrary multiplicative error factor and with an arbitrary error probability. The main ideas here are: (i)~the number of near-minimum cutsets, as opposed to the number of all cutsets, is only polynomial in the number of nodes, and all such cutsets can be enumerated in polynomial time, and (ii)~for a proper choice of near-minimum cutsets, neglecting all those cutsets which are not near-minimum yields an arbitrarily small error while approximating the failure frequency. 

We compared the accuracy and the running time of the Monte Carlo simulation (MCS) and the proposed algorithm for all-terminal reliability systems (via theoretical analysis and simulation study). The comparison results confirm that the proposed algorithm achieves a significantly better tradeoff between accuracy and running time than MCS.

Unlike the special case of all-terminal reliability systems, the number of near-minimum cutsets in $k$-terminal reliability systems (for arbitrary $2\leq k<n$) generally grows exponentially with the number of nodes. One, and perhaps the most important, problem which remains open in this area of research is to design a polynomial-time (in the number of nodes) algorithm for approximating the failure probability and the failure frequency in $k$-terminal reliability systems. Some other directions for future research include deriving tighter bounds for the analysis and improving the running time of the proposed algorithms. 

\appendix

\subsection{Proof of Lemma~\ref{lem:KLME}}
Let $P(\boldsymbol{z})$ be the probability that a randomly chosen assignment is equal to $\boldsymbol{z}$. Given that $Z_j$ is chosen in Step~1, the probability of choosing $\boldsymbol{z}$ in Step~2 is equal to $P(\boldsymbol{z})/P_Z(j)$. By the linearity of expectation, it follows that 
\begin{equation}\label{eq:EPsZi}
\mathbb{E}(\pi_{s,t}| Z_j)=\sum_{\boldsymbol{z}: \boldsymbol{z}\models Z_j} \frac{Q_Z}{N(\boldsymbol{z})} \frac{P(\boldsymbol{z})}{P_Z(j)},
\end{equation} where the notation ``$\boldsymbol{z}\models Z_j$'' indicates that the summation is taken over all $\boldsymbol{z}$ satisfying $Z_j$. (Similarly, in the following, we use the notation ``$\boldsymbol{z}\models \Phi_M$'' to indicate all $\boldsymbol{z}$ satisfying $\Phi_M$.) This gives 
\begin{eqnarray*}
\mathbb{E}(\pi_{s,t}) &\stackrel{(a)}{=}& \sum_{j} \frac{P_Z(j)}{Q_Z}\sum_{\boldsymbol{z}: \boldsymbol{z}\models Z_j} \frac{Q_Z}{N(\boldsymbol{z})} \frac{P(\boldsymbol{z})}{P_Z(j)}\\ 
&\stackrel{(b)}{=}& \sum_{j} \sum_{\boldsymbol{z}: \boldsymbol{z}\models Z_j} \frac{P(\boldsymbol{z})}{N(\boldsymbol{z})}\\ 
&\stackrel{(c)}{=}& \sum_{\boldsymbol{z}: \boldsymbol{z}\models \Phi_M} \sum_{j: \boldsymbol{z}\models Z_j} \frac{P(\boldsymbol{z})}{N(\boldsymbol{z})} \\
&\stackrel{(d)}{=}& \sum_{\boldsymbol{z}: \boldsymbol{z}\models \Phi_M} \frac{P(\boldsymbol{z})}{N(\boldsymbol{z})} \sum_{j: \boldsymbol{z}\models Z_j} 1\\
&\stackrel{(e)}{=}& \sum_{\boldsymbol{z}: \boldsymbol{z}\models \Phi_M} P(\boldsymbol{z})\\
&\stackrel{(f)}{=}& \Pi_{M}, 	
\end{eqnarray*} where $(a)$ follows from the law of total expectation, i.e., \[\mathbb{E}(\pi_{s,t})=\sum_{j} \Pr\{\text{selecting } Z_j\} \cdot \mathbb{E}(\pi_{s,t}| Z_j),\] noting that the probability of selecting $Z_j$ is equal to $P_Z(j)/Q_Z$, and $\mathbb{E}(\pi_{s,t}|Z_j)$ is given by~\eqref{eq:EPsZi}; $(b)$ follows since $Q_Z$ and $P_Z(j)$ are independent of $\boldsymbol{z}$; $(c)$ follows since any $\boldsymbol{z}$ satisfying $Z_j$ (for any $j$) also satisfies $\Phi_M$; $(d)$ follows since $P(\boldsymbol{z})/N(\boldsymbol{z})$ is independent of $j$; $(e)$ follows since $\sum_{j: \boldsymbol{z}\models Z_j} 1 = N(\boldsymbol{z})$ (by definition), and $(f)$ follows from the definition of $\Pi_M$. Thus, \[\mathbb{E}(\pi_{s,t})=\Pi_M.\] Similarly, 
\begin{eqnarray*}
\mathbb{E}((\pi_{s,t})^2) &=& \sum_{j} \frac{P_Z(j)}{Q_Z}\sum_{\boldsymbol{z}: \boldsymbol{z}\models Z_j} \frac{(Q_Z)^2}{(N(\boldsymbol{z}))^2} \frac{P(\boldsymbol{z})}{P_Z(j)}\\
&=& \sum_{j} \sum_{\boldsymbol{z}: \boldsymbol{z}\models Z_j} \frac{Q_Z}{(N(\boldsymbol{z}))^2} P(\boldsymbol{z})\\
&=& \sum_{\boldsymbol{z}: \boldsymbol{z}\models \Phi_M} \frac{Q_Z}{N(\boldsymbol{z})} P(\boldsymbol{z})\\
&{\leq}& Q_Z \Pi_M,   	
\end{eqnarray*} noting that $N(\boldsymbol{z})\geq 1$ for all $\boldsymbol{z}$ satisfying $\Phi_M$. Thus,  
\begin{eqnarray*}
\mathrm{var}(\pi_{s,t})	 &=& \mathbb{E}((\pi_{s,t})^2)-(\mathbb{E}(\pi_{s,t}))^2\\ 
&\leq & Q_Z \Pi_M - (\Pi_M)^2\\
&{\leq}& (M-1)(\Pi_M)^2,
\end{eqnarray*} noting that 
\begin{align*}
Q_Z &=\sum_{j} P_Z(j) \\ &= \sum_{j} \Pr\{\text{a randomly chosen } \boldsymbol{z} \text{ satisfies } Z_j\}\\ &\leq  M \cdot\Pr\{\text{a randomly chosen } \boldsymbol{z} \text{ satisfies } \Phi_M\}\\ &= M\Pi_M.	
\end{align*} Since $\pi_t = (\sum_{s} \pi_{s,t})/S$ for all $t\in [T]$, it follows that ${\mathbb{E}(\pi_t) = \mathbb{E}(\pi_{s,t})}$ and $\mathrm{var}(\pi_t)=\mathrm{var}(\pi_{s,t})/S$. Thus,
\[\mathbb{E}(\pi_t) =\Pi_M \] and \[\mathrm{var}(\pi_t)\leq \frac{(M-1) (\Pi_M)^2}{S}\] for any $t\in [T]$. Applying the Chebychev's inequality, we get
\begin{align*}
\Pr & \left\{|\pi_{t} -\Pi_{M}| \geq \xi \Pi_{M}\right\} \\ &\leq  \Pr\left\{|\pi_t-\Pi_M|\geq \xi \sqrt{\frac{S}{M-1}\mathrm{var}(\pi_t) }\right\} \\ &\leq \frac{M-1}{\xi^2 S}\\ &\leq \frac{1}{4},
\end{align*} for the choice of $S$ in the algorithm (Step~5). Define an indicator variable $I_t$ for all $t\in [T]$ as follows: $I_t=1$ if $|\pi_t-\Pi_M|\geq \xi \Pi_M$, and $I_t=0$ otherwise. Let $I \triangleq \sum_{t=1}^{T} I_t$. Note that $\mathbb{E}(I)\leq T/4$ since $I_t=1$ with probability at most $1/4$, and $I_t=0$ otherwise. By applying Lemma~\ref{lem:CH}, we get \[\Pr\left\{I\geq (1+\epsilon)\mathbb{E}(I)\right\}\leq \exp\left(-\left(\frac{\epsilon^2}{2+\epsilon}\right)\mathbb{E}(I)\right)\] for any $\epsilon>0$. Taking \[\epsilon = \frac{T-2\mathbb{E}(I)}{2\mathbb{E}(I)},\] we have $(1+\epsilon)\mathbb{E}(I) = T/2$. Since $\mathbb{E}(I)\leq T/4$, then $T-2\mathbb{E}(I)\geq T/2$ and $T+2\mathbb{E}(I)\leq 3T/2$. Thus, 
\begin{align}\nonumber
\Pr & \{I\geq (1+\epsilon)\mathbb{E}(I)\} \\ \nonumber &= \Pr\left\{I\geq \frac{T}{2}\right\}\\ \nonumber &\leq \exp\left(-\left(\frac{\epsilon^2}{2+\epsilon}\right)\mathbb{E}(I)\right)\\ \nonumber &= \exp\left(-\frac{1}{2}\left(\frac{(T-2\mathbb{E}(I))^2}{T+2\mathbb{E}(I)}\right)\right)\\ \nonumber &\leq \exp\left(-\frac{T}{12}\right)\\ \label{eq:IT1} &\leq \delta,
\end{align} for the choice of $T$ in the algorithm (Step~7). Note that $\tilde{\Pi}_{M}=\mathrm{median}(\{\pi_{t}\}_{t=1}^{T})$ (by definition). If the median $\tilde{\Pi}_M$ deviates from $\Pi_M$ by more than $\xi\Pi_M$ (i.e., $|\tilde{\Pi}_{M}-\Pi_M|\geq \xi\Pi_M$), then $\pi_t$ deviates from $\Pi_M$ by more than $\xi\Pi_M$ (i.e., $|\pi_t-\Pi_M|\geq \xi\Pi_M$) for at least half of $t\in [T]$, or equivalently, $I_t=1$ for at least half of $t\in [T]$ (i.e., $I\geq T/2$). Thus, 
\begin{align}\nonumber
\Pr & \left\{|\tilde{\Pi}_{M}-\Pi_{M}|\geq \xi \Pi_{M} \right\}\\ \label{eq:IT2} &\leq \Pr\left\{I \geq \frac{T}{2}\right\}.
\end{align} By combining~\eqref{eq:IT1} and~\eqref{eq:IT2}, it is easy to see that $\tilde{\Pi}_M$ is a $(\xi,\delta)$-approximation of $\Pi_M$.




\subsection{Proof of Lemma~\ref{lem:CH}}
The following inequality is useful for the proof of the lemma. For any $x> 0$ and $0\leq \theta\leq 1$,
\begin{equation}\label{eq:BI}
x^{\theta}\leq 1+(x-1)\theta. 	
\end{equation} This inequality follows immediatley from a generalization of the Brnouolli's inequality as follows. 
\begin{lemma}[{\hspace{-0.05em}\cite[Theorem~A]{S:08}}]
For any $x>-1$ and ${0\leq\theta\leq 1}$, \[(1+x)^{\theta}\leq 1+x\theta.\]
\end{lemma} 

Taking $x = e^{z}$ for any $z$, we get $e^{z I_t}\leq 1+(e^{z}-1)I_t$ for all $t$ (by~\eqref{eq:BI}). Taking expectation from both sides of this inequality, it follows that $\mathbb{E}(e^{z I_t})\leq \mathbb{E}(1+(e^{z}-1)I_t) = 1+(e^z-1)\mathbb{E}(I_t)$. Note that $\mathbb{E}(I_t) = \mathbb{E}(I)/T$ since $I = \sum_{t=1}^{T} I_t$, and $I_1,\dots,I_T$ are identically distributed. Thus, one can see that
\begin{equation}\label{eq:Ettheta1}
\mathbb{E}(e^{z I_t})\leq 1+(e^{z}-1)\frac{\mathbb{E}(I)}{T}.
\end{equation} Similarly, it can be seen that
\begin{equation}\label{eq:Ettheta2}
\mathbb{E}(e^{-z I_t})\leq 1+(e^{-z}-1)\frac{\mathbb{E}(I)}{T}.
\end{equation} Obviously, 
\begin{equation}\label{eq:Ptheta1}
\Pr\{I\geq (1+\epsilon)\mathbb{E}(I)\} = \Pr\{e^{z I}\geq e^{z(1+\epsilon)\mathbb{E}(I)}\},
\end{equation} for any $z>0$. By the Markov's inequality, $\Pr\{e^{z I}\geq e^{z(1+\epsilon)\mathbb{E}(I)}\}\leq \mathbb{E}(e^{z I})/e^{z(1+\epsilon)\mathbb{E}(I)}$. Since $I_1,\dots,I_T$ are independent, $\mathbb{E}(e^{z I}) = \mathbb{E}(e^{z I_1})\cdots \mathbb{E}(e^{z I_T})$. By~\eqref{eq:Ettheta1}, $\mathbb{E}(e^{z I})\leq (1+(e^z-1)\mathbb{E}(I)/T)^T$. Since ${(1+x/T)^T\leq e^{x}}$ for any $x$, then $\mathbb{E}(e^{z I})\leq e^{(e^{z}-1)\mathbb{E}(I)}$. Thus, 
\begin{align}\label{eq:Prtheta1} \nonumber
\Pr\{e^{z I}\geq e^{z(1+\epsilon)\mathbb{E}(I)}\} &\leq \frac{e^{(e^{z}-1)\mathbb{E}(I)}}{e^{z(1+\epsilon)\mathbb{E}(I)}}\\ & = e^{(e^{z}-1-z(1+\epsilon))\mathbb{E}(I)}.
\end{align} Let $f(z)\triangleq e^{z}-1-z(1+\epsilon)$. Taking $z = \log(1+\epsilon)$, we get $f(z)=(1+\epsilon)(1-\log(1+\epsilon))-1$. It is easy to see that $\log(1+\epsilon)\geq 2\epsilon/(2+\epsilon)$ for any $\epsilon\geq 0$. (Letting $g(\epsilon)\triangleq \log(1+\epsilon)-2\epsilon/(2+\epsilon)$, and noting $g(0)=0$ and $g(\epsilon)$ is increasing, it follows that $g(\epsilon)\geq 0$ for any $\epsilon\geq 0$.) Thus, 
\begin{equation}\label{eq:fbound1}
f(z)\leq -\epsilon^2/(2+\epsilon)	
\end{equation}
for any $\epsilon\geq 0$. By~\eqref{eq:Ptheta1}-\eqref{eq:fbound1}, we get 
\[\Pr\{I\geq (1+\epsilon)\mathbb{E}(I)\}\leq e^{-\left(\frac{\epsilon^2}{2+\epsilon}\right)\mathbb{E}(I)}.\] Similarly as in~\eqref{eq:Ptheta1}, it is obvious that
\begin{equation}\label{eq:Ptheta2}
\Pr\{I\leq (1-\epsilon)\mathbb{E}(I)\} = \Pr\{e^{-z I}\geq e^{-z(1-\epsilon)\mathbb{E}(I)}\},	
\end{equation} for any $z>0$. Similar to~\eqref{eq:Prtheta1}, except by using~\eqref{eq:Ettheta2} instead of~\eqref{eq:Ettheta1}, it follows that 
 \begin{equation}\label{eq:Prtheta2}
\Pr\{e^{-z I}\geq e^{-z(1-\epsilon)\mathbb{E}(I)}\} \leq e^{(e^{-z}-1+z(1-\epsilon))\mathbb{E}(I)}.
\end{equation} Let $f(z)\triangleq e^{-z}-1+z(1-\epsilon)$. First, suppose that $\epsilon\geq 1$. Taking $z=\epsilon$, we get $f(z)=e^{-\epsilon}-1+\epsilon(1-\epsilon)$. Since $e^{-\epsilon}\leq 1/(1+\epsilon)$ and $\epsilon/(1+\epsilon)\geq 1/2$, then $f(z)\leq -\epsilon^2(\epsilon/(1+\epsilon))\leq -\epsilon^2/2$ for any $\epsilon\geq 1$. Next, suppose that $0\leq \epsilon<1$. Taking $z = -\log(1-\epsilon)$, we get $f(z)=-(1-\epsilon)\log(1-\epsilon)-\epsilon$. By the Taylor's expansion, $\log(1-\epsilon)=-\epsilon-\epsilon^2/2-\epsilon^3/3-\epsilon^4/4-\dots$, and consequently, $(1-\epsilon)\log(1-\epsilon)=-\epsilon+\epsilon^2/2+\epsilon^3/6+\epsilon^4/12+\dots\geq -\epsilon+\epsilon^2/2$. Then, $f(z)\leq \epsilon-\epsilon^2/2-\epsilon= -\epsilon^2/2$ for any $0\leq \epsilon<1$. By these arguments, \begin{equation}\label{eq:fbound2}
f(z)\leq -\epsilon^2/2
\end{equation}
for any $\epsilon\geq 0$. By~\eqref{eq:Ptheta2}-\eqref{eq:fbound2}, we get 
\[\Pr\{I\leq (1-\epsilon)\mathbb{E}(I)\}\leq e^{-\left(\frac{\epsilon^2}{2}\right)\mathbb{E}(I)}.\]

\subsection{Proof of Lemma~\ref{lem:AlphaMinCutsNumber}}
Fix an arbitrary $\alpha\geq 1$. Consider an application of the (non-recursive) generalized contraction algorithm. Recall that in each round of this algorithm, one component, say $i$, is chosen at random with probability of choosing component $i$ equal to $w_i/w$, and the two end-nodes of the component $i$ are merged. The algorithm continues this process until more than $\lceil 2\alpha\rceil $ nodes remain, and terminates otherwise. Once terminated, the algorithm returns a randomly selected cutset in the resulting (multi-) system. 

Assume, without loss of generality, that the weights of all components are the same and equal to $w_0$. Otherwise, we can replace each component $i$ of weight $w_i$ by $w_i/w_0$ parallel components, each of weight $w_0$, for sufficiently small $w_0$ such that $w_i/w_0$ is an integer for all $i$ (for more details, see \cite{KS:96}). Note that the minimum weight of a cutset in this new system is equal to that in the original system, whereas the minimum size of a cutset in this new system can be different from that in the original system. We denote by $w^{*}$ and $s^{*}$ ($=w^{*}/w_0$), with a slight abuse of notation, the minimum weight and the minimum size of a cutset in the new system, respectively. 


Let $\mathcal{C}$ be an arbitrary cutset of weight $\beta w^{*}$ for some $1\leq \beta\leq \alpha$. Note that $\mathcal{C}$ has $\beta s^{*}$ components. We say that \emph{$\mathcal{C}$ is hit in round $i$} if one of its components is chosen and collapsed in round $i$. Consider the round $n-r+1$ where $r$ nodes are remaining (for arbitrary $\lceil 2\alpha\rceil< r\leq n$). Since the number of components connected to each node is lower bounded by $s^{*}$ (otherwise there exists a cutset of size less than $s^{*}$), then the total number of components is lower bounded by $r s^{*}/2$ (otherwise there exists a node with less than $s^{*}$ components connected to it). Thus the probability that $\mathcal{C}$ is hit in round $n-r+1$ is upper bounded by $\beta s^{*}/(r s^{*}/2)=2\beta/r\leq 2\alpha/r$. Similarly, the probability that $\mathcal{C}$ survives (i.e., $\mathcal{C}$ is not hit in) rounds $0,1,\dots,i$, and $\mathcal{C}$ is hit in round $i+1$ is upper bounded by $\beta s^{*}/((n-i) s^{*}/2)=2\beta/(n-i)$ for any $0\leq i\leq n-\lceil 2\alpha\rceil-1$. Thus, the probability that $\mathcal{C}$ survives all rounds until more than $\lceil 2\alpha\rceil$ nodes remain is lower bounded by 
\[
\prod_{i=0}^{n-\lceil2\alpha\rceil -1} \left(1-\frac{2\alpha}{n-i}\right) < \left(\frac{2}{n}\right)^{2\alpha}.
\] (See, for more details, the proof of~\cite[Theorem~2.6]{K:01}.) 

Once the algorithm terminates, there exist $\lceil 2\alpha\rceil$ nodes in the eystem, and consequently, the number of cutsets is upper bounded by the number of bipartitions of the nodes, i.e., $2^{\lceil 2\alpha\rceil-1}-1<2^{2\alpha}$. Thus, the probability that $\mathcal{C}$ is chosen via the random selection is lower bounded by $1/2^{2\alpha}$. 

The probability that $\mathcal{C}$ is output by the algorithm is the product of two probabilities: (i) the probability that $\mathcal{C}$ survives until $\lceil 2\alpha\rceil$ nodes remain (this probability is lower bounded by $2^{2\alpha}/n^{2\alpha}$), and (ii) the probability that $\mathcal{C}$ is chosen by the random selection (this probability is lower bounded by $1/2^{2\alpha}$). Thus the probability that $\mathcal{C}$ is output by the algorithm is lower bounded by $(2^{2\alpha}/n^{2\alpha}) \times (1/2^{2\alpha})= 1/n^{2\alpha}$. 

Each run of the algorithm outputs an $\alpha$-min cutset with probability at least $1/n^{2\alpha}$. Thus, it follows that the number of $\alpha$-min cutsets is at most $n^{2\alpha}$. (This holds because the algorithm returns each distinct $\alpha$-min cutset with probability lower bounded by $1/n^{2\alpha}$, and such events are disjoint.) 

Thus $O(n^2)$ runs of the contraction algorithm are sufficient to find a min-cutset with high probability. By a clever recursive implementation of the contraction algorithm, as shown in \cite[Lemma 4.1]{KS:96} and \cite[Lemma 4.3]{KS:96}, a min-cutset can be found in $O(n^{2} \log^2 n)$ time (instead of $O(n^4)$ time for the obvious implementation) with high probability.

\subsection{Proof of Lemma~\ref{lem:AlphaMinCutsetsEnum}}
The proof follows from the \emph{coupon-collector} argument~\cite{ER:61}: if there are $M$ bins, and potentially an infinite number of balls to be thrown independently and uniformly one at a time, then throwing $M\log (M/\delta)$ balls suffices with probability at least $1-\delta$ to have each bin contain at least one ball. (To be specific, throwing $M\log M - M\log\log (1/(1-\delta))$ balls suffices with probability $1-\delta$ as $M$ grows large~\cite{ER:61}.) By the result of Lemma~\ref{lem:AlphaMinCutsNumber}, the number of $\alpha$-min cutsets is at most $n^{2\alpha}$. Think of $\alpha$-min cutsets as bins and the runs of the (non-recursive) generalized contraction algorithm as balls. By the coupon-collector argument, one can enumerate all $\alpha$-min cutsets with probability at least $1-n^{-c}$, for any $c>0$, in $O(n^{2\alpha+2}\log(n^{2\alpha+c}))$ time  (by running the generalized contraction algorithm $n^{2\alpha}\log (n^{2\alpha+c})$ times). By running RGC algorithm $n^{2\alpha}\log (n^{2\alpha+c})$ times, as was shown in \cite[Lemma 4.1]{KS:96} and \cite[Lemma 4.3]{KS:96}, all $\alpha$-min cutsets can be found in $O(n^{2\alpha}(\log (n^{2\alpha+c}))(\log n)) = O(n^{2\alpha} \log^2 n)$ time with probability at least $1-n^{-c}$ (see, e.g.,~\cite[Theorem 8.5]{KS:96}). 


\bibliographystyle{IEEEtran}
\bibliography{IEEEabrv,FailureRefs}

\begin{thebibliography}{10}
\providecommand{\url}[1]{#1}
\csname url@samestyle\endcsname
\providecommand{\newblock}{\relax}
\providecommand{\bibinfo}[2]{#2}
\providecommand{\BIBentrySTDinterwordspacing}{\spaceskip=0pt\relax}
\providecommand{\BIBentryALTinterwordstretchfactor}{4}
\providecommand{\BIBentryALTinterwordspacing}{\spaceskip=\fontdimen2\font plus
\BIBentryALTinterwordstretchfactor\fontdimen3\font minus
  \fontdimen4\font\relax}
\providecommand{\BIBforeignlanguage}[2]{{%
\expandafter\ifx\csname l@#1\endcsname\relax
\typeout{** WARNING: IEEEtran.bst: No hyphenation pattern has been}%
\typeout{** loaded for the language `#1'. Using the pattern for}%
\typeout{** the default language instead.}%
\else
\language=\csname l@#1\endcsname
\fi
#2}}
\providecommand{\BIBdecl}{\relax}
\BIBdecl

\bibitem{GMP:64}
D.~P. Gaver, F.~E. Montmeat, and A.~D. Patton, ``Power systems reliability
  {I}-measures of reliability and methods of calculation,'' \emph{IEEE Trans.
  Power App. Syst.}, vol.~83, no.~7, pp. 727--737, Jul. 1964.

\bibitem{SB:77}
C.~Singh and R.~Billinton, \emph{System Reliability Modeling and
  Evaluation}.\hskip 1em plus 0.5em minus 0.4em\relax London, U.K.: Hutchinson
  Educational, 1977.

\bibitem{HTL:81}
C.~L. Hwang, F.~K. Tillman, and M.~H. Lee, ``System-reliability evaluation
  techniques for complex/large systems--{A} review,'' \emph{IEEE Trans.
  Reliab.}, vol.~30, no.~5, pp. 416--423, Dec. 1981.

\bibitem{BB:68}
R.~Billinton and K.~E. Bollinger, ``Transmission system reliability evaluation
  using {M}arkov processes,'' \emph{IEEE Trans. Power App. Syst.}, vol. PAS-87,
  no.~2, pp. 538--547, Feb. 1968.

\bibitem{A:79}
J.~A. Abraham, ``An improved algorithm for network reliability,'' \emph{IEEE
  Trans. Reliab.}, vol. R-38, no.~1, pp. 58--61, 1979.

\bibitem{S:81}
C.~Singh, ``Markov cut-set approach for the reliability evaluation of
  transmission and {D}istribution {S}ystems,'' \emph{IEEE Trans. Power App.
  Syst.}, vol. PAS-100, no.~6, pp. 2719--2725, Jun. 1981.

\bibitem{F:86}
G.~S. Fishman, ``A {M}onte {C}arlo sampling plan for estimating network
  reliability,'' \emph{Operations Research}, vol.~34, no.~4, pp. 581--594,
  1986.

\bibitem{L:87}
M.~O. Locks, ``A minimizing algorithm for sum of disjoint products,''
  \emph{IEEE Trans. Reliab.}, vol. R-36, no.~4, pp. 445--453, 1987.

\bibitem{BS:87}
F.~Beichelt and L.~Spross, ``An improved {A}braham-method for generating
  disjoint sums,'' \emph{IEEE Trans. Reliab.}, vol. R-36, no.~1, pp. 70--74,
  Apr. 1987.

\bibitem{DS:88}
C.~Dichirico and C.~Singh, ``Reliability analysis of transmission lines with
  common mode failures when repair times are arbitrarily distributed,''
  \emph{IEEE Trans. Power Syst.}, vol.~3, no.~3, pp. 1012--1019, Aug. 1988.

\bibitem{H:89}
K.~D. Heidtmann, ``Smaller sums of disjoint products by subproduct inversion,''
  \emph{IEEE Trans. Reliab.}, vol.~38, no.~3, pp. 305--311, 1989.

\bibitem{W:90}
J.~M. Wilson, ``An improved minimizing algorithm for sum of disjoint
  products,'' \emph{IEEE Trans. Reliab.}, vol.~39, no.~1, pp. 42--45, Apr.
  1990.

\bibitem{LW:92}
M.~O. Locks and J.~M. Wilson, ``Note on disjoint algorithms,'' \emph{IEEE
  Trans. Reliab.}, vol.~41, no.~3, pp. 81--92, Mar. 1992.

\bibitem{LJY:95}
J.~Lin, C.~Jane, and J.~Yuan, ``On reliability evaluation of a capacitated-flow
  network in terms of minimal path sets,'' \emph{Networks}, vol.~25, no.~3, pp.
  131--138, May 1995.

\bibitem{LP:01}
S.~M. Lee and D.~H. Park, ``An efficient method for evaluating network
  reliability with variable link-capacities,'' \emph{IEEE Trans. Reliab.},
  vol.~50, no.~4, pp. 374--379, Dec. 2001.

\bibitem{BT:03}
A.~Balan and L.~Traldi, ``Preprocessing minpaths for sum of disjoint
  products,'' \emph{IEEE Trans. Reliab.}, vol.~52, no.~3, pp. 289--295, Sep.
  2003.

\bibitem{SB:74}
C.~Singh and R.~Billinton, ``A new method to determine the failure frequency of
  a complex system,'' \emph{IEEE Trans. Reliab.}, vol.~23, no.~4, pp. 231--234,
  Oct. 1974.

\bibitem{S:79}
C.~Singh, ``A matrix approach to calculate the failure frequency and related
  indices,'' \emph{Microelectron. Reliab.}, vol.~19, no.~4, pp. 395--398, 1979.

\bibitem{S:80}
------, ``Effect of probability distributions on steady state frequency,''
  \emph{IEEE Trans. Reliab.}, vol.~29, no.~3, p. 274, 1980.

\bibitem{S2:81}
------, ``Rules for calculating the time-specific frequency of system
  failure,'' \emph{IEEE Trans. Reliab.}, vol.~30, no.~4, pp. 364--366, Oct.
  1981.

\bibitem{B:82}
J.~V. Bukowski, ``On the determination of large-scale system reliability,''
  \emph{IEEE Trans. Systems, Man, and Cybernetics}, vol.~12, no.~4, pp.
  538--548, Jul. 1982.

\bibitem{PB:83}
J.~Provan and M.~Ball, ``The complexity of counting cuts and of computing the
  probability that a graph is connected,'' \emph{SIAM Journal on Computing},
  vol.~12, no.~4, pp. 777--788, 1983.

\bibitem{B:86}
M.~Ball, ``Computational complexity of network reliability analysis,''
  \emph{IEEE Trans. Reliab.}, vol.~35, no.~3, pp. 230--239, Aug. 1986.

\bibitem{MPL:94}
J.~C.~O. Mello, M.~V.~F. Pereira, and A.~M. Leite~da Silva, ``Evaluation of
  reliability worth in composite systems based on pseudo-sequencial {M}onte
  {C}arlo simulation,'' \emph{IEEE Trans. Power Syst.}, vol.~9, no.~3, pp.
  1318--1326, Aug. 1994.

\bibitem{LCFA:02}
A.~M. Leite~da Silva, J.~G. de~Carvalho~Costa, L.~A. da~Fonseca~Manso, and
  G.~J. Anders, ``Transmission capacity: Availability, maximum transfer and
  reliability,'' \emph{IEEE Trans. Power Syst.}, vol.~17, no.~3, pp. 843--849,
  Aug. 2002.

\bibitem{SKS:02}
C.~Srivareeratana, A.~Konak, and A.~E. Smith, ``Estimation of all-terminal
  network reliability using an artificial neural network,'' \emph{Computers and
  Operations Research}, vol.~29, no.~7, pp. 849--868, 2002.

\bibitem{HBKK:05}
K.-P. Hui, N.~Bean, M.~Kraetzl, and D.~Kroese, ``The cross-entropy method for
  network reliability estimation,'' \emph{Operations Research}, vol. 134,
  no.~1, pp. 101--118, 2005.

\bibitem{ADS:09}
F.~Altiparmak, B.~Dengiz, and A.~E. Smith, ``A general neural network model for
  estimating telecommunications network reliability,'' \emph{IEEE Trans.
  Reliab.}, vol.~58, no.~1, pp. 2--9, Mar. 2009.

\bibitem{JB:69}
P.~A. Jensen and M.~Bellmore, ``An algorithm to determine the reliability of a
  complex system,'' \emph{IEEE Trans. Reliab.}, vol.~18, no.~4, pp. 169--174,
  Nov. 1969.

\bibitem{NBB:70}
A.~C. Nelson, J.~R. Batts, and R.~L. Beadles, ``A computer program for
  approximating system reliability,'' \emph{IEEE Trans. Reliab.}, vol.~19,
  no.~2, pp. 61--65, May 1970.

\bibitem{EP:70}
J.~D. Esary and F.~Proschan, ``A reliability bound for systems of maintained,
  independent components,'' \emph{Journal of the American Statistical
  Association}, vol.~65, pp. 329--338, 1970.

\bibitem{WK:11}
J.-M. Won and F.~Karray, ``A greedy algorithm for faster feasibility evaluation
  of all-terminal-reliable networks,'' \emph{IEEE Trans. Systems, Man, and
  Cybernetics}, vol.~41, no.~6, pp. 1600--1611, 2011.

\bibitem{S:77}
C.~Singh, ``On the behavior of failure frequency bounds,'' \emph{IEEE Trans.
  Reliab.}, vol.~26, no.~1, pp. 63--66, Apr. 1977.

\bibitem{MS:99}
J.~Mitra and C.~Singh, ``Pruning and simulation for determination of frequency
  and duration indices of composite power systems,'' \emph{IEEE Trans. Power
  Syst.}, vol.~14, no.~3, pp. 899--905, Aug. 1999.

\bibitem{K:01}
D.~R. Karger, ``A randomized fully polynomial time approximation scheme for the
  all terminal network reliability problem,'' \emph{SIAM Review}, vol.~43,
  no.~3, pp. 499--522, 2001.

\bibitem{C:87}
C.~J. Colbourn, \emph{The Combinatorics of Network Reliability}.\hskip 1em plus
  0.5em minus 0.4em\relax New York, NY, USA: Oxford University Press, Inc.,
  1987.

\bibitem{SM:09}
A.~R. Sharafat and O.~R. Ma'rouzi, ``All-terminal network reliability using
  recursive truncation algorithm,'' \emph{IEEE Trans. Reliab.}, vol.~58, no.~2,
  pp. 338--347, June 2009.

\bibitem{I2net}
``https://www.internet2.edu/products-services/advanced-networking/.''

\bibitem{Sn:79}
C.~Singh, ``Calculating the time-specific frequency of system failure,''
  \emph{IEEE Trans. Reliab.}, vol. R-28, no.~2, pp. 124--126, Jun. 1979.

\bibitem{TSOA:80}
S.~Tsukiyama, I.~Shirakawa, H.~Ozaki, and H.~Ariyoshi, ``An algorithm to
  enumerate all cutsets of a graph in linear time per cutset,'' \emph{J. ACM},
  vol.~27, no.~4, pp. 619--632, Oct. 1980.

\bibitem{KLM:89}
R.~M. Karp, M.~Luby, and N.~Madras, ``Monte {C}arlo approximation algorithms
  for enumeration problems,'' \emph{J. Algorithms}, vol.~10, no.~3, pp.
  429--448, Sep. 1989.

\bibitem{KS:96}
D.~R. Karger and C.~Stein, ``A new approach to the minimum cut problem,''
  \emph{J. ACM}, vol.~43, no.~4, pp. 601--640, Jul. 1996.

\bibitem{K:93}
D.~R. Karger, ``Global min-cuts in {RNC}, and other ramifications of a simple
  min-out algorithm,'' in \emph{Proceedings of the Fourth Annual ACM-SIAM
  Symposium on Discrete Algorithms}, ser. SODA '93, 1993, pp. 21--30.

\bibitem{SW:94}
M.~Stoer and F.~Wagner, \emph{A simple min cut algorithm}.\hskip 1em plus 0.5em
  minus 0.4em\relax Berlin, Heidelberg: Springer Berlin Heidelberg, 1994, pp.
  141--147.

\bibitem{S:08}
H.-n. Shi, ``Generalizations of {B}ernoulli's inequality with applications,''
  \emph{Journal of Mathematical Inequalities}, vol.~2, Jan. 2008.

\bibitem{ER:61}
P.~Erd{\"o}s and A.~R{\'e}nyi, ``On a classical problem of probability
  theory,'' \emph{Publ. Math. Inst. Hung. Acad. Sci.}, pp. 215--220, 1961.

\end{thebibliography}

\end{document}